\documentclass[11pt]{article}

\usepackage{verbatim}
\usepackage{graphicx}
\usepackage{tabularx}

\usepackage{amsmath}
\usepackage{bbm}
\usepackage[showonlyrefs]{mathtools}
\usepackage{amstext,amssymb,amsfonts}
\usepackage{bm}

\usepackage{algorithmic,algorithm}

\DeclarePairedDelimiter\ceil{\lceil}{\rceil}
\DeclarePairedDelimiter\floor{\lfloor}{\rfloor}

\usepackage[top=1in, bottom=1in, left=1in, right=1in]{geometry}

\usepackage{nameref}
\usepackage[pagebackref,colorlinks,linkcolor=blue,filecolor = blue,citecolor = blue, urlcolor = blue, hyperfootnotes=false]{hyperref}
\usepackage{color}

\usepackage{amsthm}
\usepackage{thmtools}

\declaretheorem[within=section]{theorem}
\declaretheorem[sibling=theorem]{corollary}

\declaretheorem[sibling=theorem]{lemma}
\declaretheorem[sibling=theorem]{proposition}
\declaretheorem[sibling=theorem]{claim}

\declaretheorem[sibling=theorem]{definition}

\newcommand{\ignore}[1]{}

\newcommand{\inpro}[2]{\left\langle #1,#2 \right\rangle}

\newcommand{\linearspan}[1]{\mathrm{span}\{#1\}}
\newcommand{\img}{\mathrm{Im}}

\newcommand{\set}[1]{\{#1\}}

\newcommand{\rk}{\mathrm{rank}}

\newcommand{\sgn}{\mathrm{sgn}}

\renewcommand{\leq}{\leqslant}
\renewcommand{\geq}{\geqslant}
\renewcommand{\ge}{\geqslant}
\renewcommand{\le}{\leqslant}
\renewcommand{\epsilon}{\varepsilon}
\newcommand{\eps}{\epsilon}

\newcommand{\Z}{\mathbb{Z}}

\newcommand{\F}{\mathbb{F}}

\newcommand{\K}{\mathbb{K}}
\newcommand{\PF}{\mathbb{P}\mathbb{F}}

\newcommand{\cO}{\mathcal O}

\newcommand{\Esymb}{{\bf E}}

\newcommand{\Psymb}{{\bf Pr}}

\DeclareMathOperator*{\E}{\Esymb}

\DeclareMathOperator*{\ProbOp}{\Psymb}

\renewcommand{\Pr}{\ProbOp}

\usepackage{libertine}
\usepackage[libertine]{newtxmath}
\def\notes{1}
 \newcommand{\snote}[1]{\ifnum\notes=1{{\sf\color{green} [Sergey: #1]}}\fi}
 \newcommand{\vnote}[1]{\ifnum\notes=1{{\sf\color{blue} [Venkat: #1]}}\fi}
 \newcommand{\gnote}[1]{\ifnum\notes=1{{\sf\color{red} [Gopi: #1]}}\fi}

\newcommand{\mattwoone}[2]{
\left(
\begin{matrix}
#1\\
#2\\
\end{matrix}
\right)}

\newcommand{\matthreeone}[3]{
\left(
\begin{matrix}
#1\\
#2\\
#3\\
\end{matrix}
\right)}

\def\PK{{\mathbb{P}\mathbb{K}}}
\def\bK{{\overline{\mathbb{K}}}}
\def\PbK{{\mathbb{P}\overline{\mathbb{K}}}}

\def\Li{{\mathrm{Li}}}

\begin{document}

\title{Maximally Recoverable LRCs: A field size lower bound and constructions for few heavy parities}

\author{
Sivakanth Gopi\thanks{Microsoft Research. Email: \texttt{sigopi@microsoft.com}. Research supported by NSF CAREER award 1451191 and NSF grant CCF-1523816. Most of this work was done when the author was visiting Microsoft Research in Summer 2017.}\\
\and    Venkatesan Guruswami\thanks{Carnegie Mellon University. Email: \texttt{venkatg@cs.cmu.edu}. Research supported in part by NSF grant CCF-1563742. Most of this work was done during a visit by the author to Microsoft Research, Redmond. The work was also partly done when
the author was visiting the School of Physical and Mathematical Sciences, Nanyang Technological University, Singapore, and the Center of Mathematical Sciences and Applications, Harvard University.}  \\
\and    Sergey Yekhanin\thanks{Microsoft Research. Email: \texttt{yekhanin@microsoft.com} }
}

\date{}

\maketitle

\begin{abstract} 
The explosion in the volumes of data being stored online has resulted in distributed storage systems transitioning to erasure coding based schemes. Local Reconstruction Codes (LRCs) have emerged as the codes of choice for these applications. These codes can correct a small number of erasures (which is the typical case) by accessing only a small number of remaining coordinates. An $(n,r,h,a,q)$-LRC is a linear code over $\mathbb{F}_q$ of length $n$, whose codeword symbols are partitioned into $g=n/r$ local groups each of size $r$. Each local group has $a$ local parity checks that allow recovery of up to $a$ erasures within the group by reading the unerased symbols in the group. There are a further $h$ ``heavy" parity checks to provide fault tolerance from more global erasure patterns. Such an LRC is Maximally Recoverable (MR), if it corrects all erasure patterns which are information-theoretically correctable under the stipulated structure of local and global parity checks, namely patterns with up to $a$ erasures in each local group and an additional $h$ (or fewer) erasures anywhere in the codeword. 

The existing constructions require fields of size $n^{\Omega(h)}$ while no superlinear lower bounds were known for any setting of parameters. Is it possible to get linear field size similar to the related MDS codes (e.g. Reed-Solomon codes)? In this work, we answer this question by showing superlinear lower bounds on the field size of MR LRCs. When $a,h$ are constant and the number of local groups $g \ge h$, while $r$ may grow with $n$, our lower bound simplifies to 
$$q\ge \Omega_{a,h}\left(n\cdot r^{\min\{a,h-2\}}\right).$$ 
 
MR LRCs deployed in practice have a small number of global parities, typically $h=2,3$~\cite{HuangSX}. We complement our lower bounds by giving constructions with small field size for $h\le 3$.  When $h=2$, we give a linear field size construction, whereas previous constructions required quadratic field size in some parameter ranges. Note that our lower bound is superlinear only if $h\ge 3$. When $h=3$, we give a construction with $O(n^3)$ field size, whereas previous constructions needed $n^{\Theta(a)}$ field size. Our construction for $h=2$ makes the choices $r=3, a=1, h=3$ the next smallest setting to investigate regarding the existence of MR LRCs over fields of near-linear size. We answer this question in the positive via a novel approach based on elliptic curves and arithmetic progression free sets.
\end{abstract}
\setcounter{page}{0}
\thispagestyle{empty}
\newpage

\section{Introduction}\label{Sec:Intro}
The explosion in the volumes of data being stored online means that duplicating or triplicating data is not economically feasible. This has resulted in distributed storage systems employing erasure coding based schemes in order to ensure reliability with low storage overheads. In recent years Local Reconstruction Codes (LRCs) emerged as the codes of choice for many such scenarios and have been implemented in a number of large scale systems e.g., Microsoft Azure~\cite{HuangSX} and Hadoop~\cite{XOR_ELE}.

Classical erasure correcting codes~\cite{MS} guarantee that data can be recovered if a bounded number of codeword coordinates is erased. However recovering data typically involves accessing all surviving coordinates. By contrast, Local Reconstruction Codes\footnote{The term local reconstruction codes is from~\cite{HuangSX}. Essentially the same codes were called locally repairable codes in~\cite{Dimakis_0} and locally recoverable codes in~\cite{TB}. Thankfully all names above abbreviate to LRCs.} (LRCs) distinguish between the typical case when only a small number of codeword coordinates are erased (e.g., few machines in a data center fail) and a worst case when a larger number of coordinates might be unavailable, and guarantee that in the prior case recovery of individual coordinates can be accomplished in sub-linear time, without having to access all surviving symbols.

LRCs are systematic linear codes, where encoding is a two stage process. In the first stage, $h$ redundant heavy parity symbols are generated from $k$ data symbols. Each heavy parity is a linear combination of all $k$ data symbols. During the second stage, the $k+h$ symbols are partitioned into $\frac{k+h}{r-a}$ sets of size $r-a$ and each set is extended with $a$ local parity symbols using an MDS code to form a \emph{local group} as shown in Figure~\ref{Fig:LRC}. Encoding as above ensures that when at most $a$ coordinates are erased, any missing coordinate can be recovered by accessing at most $r-a$ symbols. However, if a larger number of coordinates (that depends on $h$) is erased; then all missing symbols can be recovered by potentially accessing all remaining symbols.

\begin{figure}[h]
\label{Fig:LRC}
\includegraphics[scale=0.45,trim={0cm 13cm 19cm 10cm},clip]{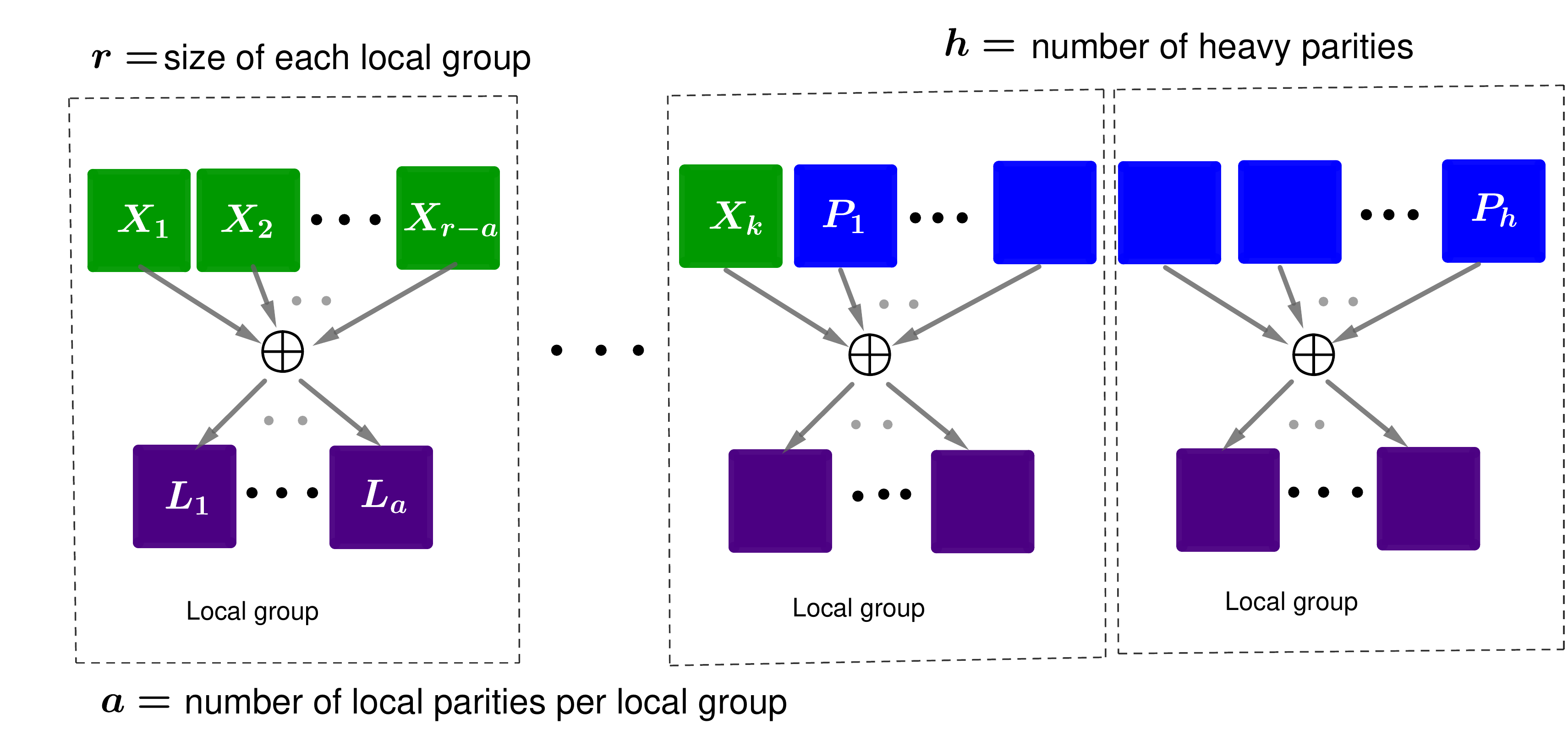}
\caption{An LRC with $k$ data symbols, $h$ heavy parities and `$a$' local parities per local group.}
\end{figure}

Our description of LRC codes above is not complete. To specify a concrete code we need to fix coefficients in linear combinations that define $h$ heavy and $\frac{k+h}{r-a}\cdot a$ local parities. Different choices of coefficients could lead to codes with different erasure correcting capabilities. The best we could hope for is to have an optimal choice of coefficients which ensures that our code can correct every pattern of erasures that is correctable for some setting of coefficients. Such codes always exist and are called Maximally Recoverable (MR)~\cite{CHL,HCL} LRCs.\footnote{Maximally recoverable LRCs are called Partial MDS (PMDS) in~\cite{Blaum,BHH} and many follow up works.} Combinatorially, an $(n,r,h,a,q)$-LRC is maximally recoverable it if corrects every pattern of erasures that can be obtained by erasing $a$ coordinates in each local group and up to $h$ additional coordinates elsewhere, here $q$ is the size of the field over which the linear code is defined. Explicit constructions of MR LRCs are available (e.g.,~\cite{CK}) for all ranges of parameters. Unfortunately, all known constructions require finite fields of very large size.

Encoding  a linear code and decoding it from erasures involve matrix vector multiplication and linear equation solving respectively. Both of these require performing numerous finite field arithmetic operations. Having small finite fields results in faster encoding and decoding and thus improves the overall throughput of the system~\cite[Section 2]{Plank}. It is also desirable in practice to work over finite fields of characteristic 2. Obtaining MR LRCs over finite fields of minimal size is one of the central problems in the area of codes for distributed storage.

\subsection{State of the art and our results}\label{SubSec:Results}

We now summarize what is known about the minimal field size of maximally recoverable local reconstruction codes with parameters $n,r,a$ and $h$ and first cover the easy cases.
\begin{itemize}
  \itemsep=0ex
\item When $a=0,$ LRCs are equivalent to classical erasure correcting codes. In this case Reed Solomon codes are maximally recoverable, and they have a field size of roughly $n,$ which is known to be optimal up to constant factors~\cite{MainMDS1}.
\item When $h=0$ or $h=1$, there are constructions of maximally recoverable LRCs over fields of size $O(r)$~\cite{BHH} which is optimal.
\item When $r=a+1,$ codes in the local groups are necessarily simple repetition codes. MR LRCs can be obtained by starting with a Reed Solomon code of length $n/r$  and repeating every coordinate $r$ times. Thus the optimal field size is $\Theta(n/r).$
\end{itemize}
This leaves us with the main case, when $a\geq 1,$ $r\geq a+2,$ and $h\geq 2.$ A number of constructions have been obtained~\cite{Blaum,BHH,TPD,GHJY,HY,GHKSWY,CK,BPSY,GYBS}. The best constructions for the case of $h=2$ are from~\cite{BPSY} and require a field of size $O(a\cdot n).$ For most other settings of parameters the best families of MR LRCs are from~\cite{GYBS}. They present two different constructions with field size
\begin{equation}
\label{Eqn:StateOfArt}
\begin{aligned}
O\left(r\cdot n^{(a+1)h-1}\right) \quad\mathrm{and} \quad O\left(\max\left(O(n/r), O(r)^{h+a}\right)^h\right)
\end{aligned}
\end{equation}
respectively.
The first bound is typically better when $r=\Omega(n).$ The second bound is better when $r\ll n.$   A recent (unpublished) work~\cite{GLX-ff} uses a new approach based on function fields to obtain some improvements to the above bounds in certain cases (e.g. $a=1$, or when $r$ is small and $h$ is large), but the exponential dependence on $h$ in the field size remains. Thus in all known constructions, the field size $q$ grows rapidly with the codeword length. With this context, we are now ready to discuss our results.

\medskip \noindent \textbf{Lower bound.}
The bounds in (\ref{Eqn:StateOfArt}) exhibit code constructions but not any inherent limitations. In particular, up until our work it remained a possibility that codes over fields of size $O(n)$ could exist for all ranges of LRC parameters. We obtain the first superlinear lower bound on the field size of MR LRCs, prior to our work no superlinear lower bounds were known in any setting of parameters.
\begin{theorem}\label{Th:lowerbound_mega}
Let $h\ge 2$ and $a$ be fixed constants while $r$ may grow with $n$. Any maximally recoverable $(n,r,h,a,q)$-LRC with $g=n/r\ge 2$ local groups must have:
   \begin{equation}\label{Eqn:MainLowerFormOur_mega}
q \ge \Omega_{h,a} \left( n\cdot r^\alpha \right) \text{ where } \alpha=\frac{\min\left\{a,h-2\ceil{h/g}\right\}}{\ceil{h/g}}.
\end{equation}
\end{theorem}
\noindent
The lower bound~(\ref{Eqn:MainLowerFormOur_mega}) simplifies as follows in some special cases:
\begin{itemize}
\itemsep=-0.25ex
\item $g\ge h$ : $q\ge \Omega_{h,a}\left(nr^{\min\{a,h-2\}}\right)$
\item $g\le h$, $g$ divides $h$ and $a\le h-2h/g$ : $q\ge \Omega_{h,a}\left(n^{1+ag/h}\right)$
\item $g\le h$, $g$ divides $h$ and $a> h-2h/g$ : $q\ge \Omega_{h,a}\left(n^{g-1}\right)$.
\end{itemize}


Note that our lower bound is superlinear whenever $r$ is growing with $n$ except when $a=0$ or $h=2$ or $g=2$ or $(g=3, h=4,a=1)$. 
We believe that from a practical standpoint, the setting of $r$ slowly growing with $n$ (like say $r = \log n$ or $r = n^\eps$) is interesting because if $r$ is constant, the number of parity checks or redundant symbols $(an/r + h)$ will be linear in $n$, and applications of codes in distributed storage demand high rate codes.


When $a=0$, MR LRCs reduce to MDS codes and so there are linear field size constructions (Reed-Solomon codes). When $h=2$, we obtain a linear field size construction (Theorem~\ref{Th:H2Main}). This leaves $g=2$ and $(g=3, h=4,a=1)$  as the only cases where we don't know if linear field size is enough for MR LRCs.

The parity check view of MR LRCs throws a different light on our lower bound. The parity check matrix of an MR $(n,r,h,a,q)$-LRC with $g=n/r$ local groups is an $(ag+h)\times n$ matrix of the following form:
\begin{equation}\label{fig:MRtopology_intro}
H=
\left[
\begin{array}{c|c|c|c}
A_1 & 0 & \cdots & 0\\
\hline
0 &A_2 & \cdots & 0\\
\hline
\vdots & \vdots & \ddots & \vdots \\
\hline
0 & 0 & \cdots & A_g \\
\hline
B_1 & B_2 & \cdots & B_g \\
\end{array}
\right].
\end{equation}
Here $A_1,A_2,\cdots,A_g$ are $a\times r$ matrices over $\F_q$, $B_1,B_2,\cdots,B_g$ are $h\times r$ matrices over $\F_q.$ The rest of the matrix is filled with zeros. An erasure pattern with $ag+h$ erasures is correctable iff the corresponding minor in $H$ is non-zero. Thinking of the entries of the matrices $A_i,B_i$ as variables, every $(ag+h)\times (ag+h)$ minor of $H$ is either identically zero or a non-zero polynomial in those variables. We call the zero minors as trivial and the rest as non-trivial. It turns out that the non-trivial minors of $H$ in (\ref{fig:MRtopology_intro}) are exactly those which are obtainable by selecting $a$ columns in each local group and $h$ additional columns anywhere. There exists an MR LRC over $\F_q$ with these parameters iff there exists an assignment of $\F_q$ values to these variables which makes all the non-trivial minors non-zero. 
It is easy to see that if we assign random values from a large enough finite field $\F_q$ (say $q\gg n^{ag+h}$) to the variables, by Schwartz-Zippel lemma, all the non-trivial minors will be non-zero with high probability. But this probabilistic argument can only work for very large fields. Seen this way, it seems very natural to ask what is the smallest field size required to make all the non-trivial minors non-zero given a matrix with some pattern of zeros.

Thus our lower bound shows that one needs super linear size fields to instantiate $H$ to make all non-trivial minors non-zero. This is even more surprising when contrasted with a recent proof of the GM-MDS conjecture by Lovett~\cite{Lovett18} and independently by Yildiz and Hassibi~\cite{YH18}. This states that a $k\times n$ matrix ($k\le n$) with some pattern of zeros such that every $k\times k$ minor is non-trivial can be instantiated with a field of size $q\le n+k-1$ to make every $k\times k$ minor non-zero.  
 
 \medskip \noindent \textbf{Upper bounds (Code constructions).}
MR LRCs that are deployed in practice typically have a small constant number of global parities, typicially $h=2,3$~\cite{HuangSX}. Without explicit constructions, one has to search over assignments from a small field to variables in the parity check matrix~(\ref{fig:MRtopology_intro}) to find an assignment which makes all the non-trivial minors non-zero. This is prohibitively expensive even for small values of $n$ and $q$ that are deployed in practice. Note that for random assignments to work with high probability, the field should be very large. Keeping this in mind, we design explicit MR LRCs over small field size for $h\le 3$.
 \begin{itemize}
 \itemsep=0ex
\item We obtain a family of MR $(n,r,h=2,a,q)$-LRCs, where $q=O(n)$ for all settings of parameters. Prior to our work the best constructions~\cite{BPSY} required $q$ to be $O(a\cdot n)$ which in general may be up to quadratic in $n.$ If we require that the field has characteristic two, we get such codes with $q=n^{1+o(1)}$.
\item  We obtain a family of MR $(n,r,h=3,a,q)$-LRCs, where $q=O(n^{3})$ for all settings of parameters. Prior to our work the best constructions~(\ref{Eqn:StateOfArt}) required $q$ to be up to $n^{\Theta(a)}$ for some regimes.  If we require that the field has characteristic two, we can get such codes with $q=n^{3+o(1)}$.
\item Given our linear field size construction for $h=2$ (and since the problem is trivial for $r=2$), the setting $r=3, a=1, h=3$ is the next smallest regime to investigate regarding the existence of MR LRCs over fields of near-linear size. We construct such MR LRCs with a field size of $n \cdot \exp(O(\sqrt{\log n}))$ by developing a new approach to LRC constructions based on elliptic curves and AP-free sets.
\end{itemize}

\subsection{Our techniques}\label{SubSec:Techniques}
Similar to most earlier works in the area we represent LRC codes via their parity check matrices which look like~(\ref{fig:MRtopology_intro}). Such matrices $H$ have size $\left(a\cdot g+h\right) \times n$ and a simple block structure. Columns are partitioned into $r$-sized local groups. For each local group there is a corresponding collection of $a$ rows that impose $MDS$ constraints on coordinates in the group, and have no support outside the group. Remaining $h$ rows of $H$ correspond to heavy parity symbols and carry arbitrary values.

To establish our lower bound when $g\ge h$, we start with a parity check matrix of an arbitrary maximally recoverable local reconstruction code. From it, we obtain a family of large mutually disjoint subsets $X_1,\ldots,X_g$ in the projective space $\PF_q^{h-1},$ such that no hyperplane in $\PF_q^{h-1}$ intersects $h$ distinct sets among $X_1,\dots,X_g$. For example when $a=1$ and $h\ge 3$, the set $X_i$ is all the pairwise differences of columns of $B_i$ in ~(\ref{fig:MRtopology_intro}) thought of as points in $\PF_q^{h-1}$. We then show that if $q$ is too small, then a random hyperplane will intersect $h$ distinct sets among $X_1,\dots,X_g$ with positive probability, which gives the required lower bound. When $h>g$, each $X_i$ will be a collection of subspaces in $\F_q^h$ of dimension roughly $h/g$ such that any collection of $g$ subspaces, one from each $X_i$, will span $\F_q^h$. Again we show that if $q$ is too small, a random $(h-1)$-dimensional subspace will contain a subspace each from $X_i$ with high probability. The proof is more intricate in this case, because we need to carefully calculate how subspaces inside each $X_i$ intersect with each other.

We now explain the main ideas behind our constructions. An LRC is MR if any subset of columns of $H$ (as in~(\ref{fig:MRtopology_intro})) that can be obtained by selecting $a$ columns from each local group and then $h$ more has full rank. Suppose all $h$ additional columns are selected from distinct local groups. In this case showing that some $ag+h$ columns are independent easily reduces to showing that a certain $(ah+h)\times (ah+h)$ determinant is non-zero. An important algebraic identity that underlies our constructions for $h=2$ and $h=3$ reduces such determinants to much smaller $h\times h$ determinants of determinants in the entries of $H.$ A special case of this identity when $h=2$ and matrices are Vandermonde type appears in~\cite{BPSY}. In addition to that we utilize various properties of finite fields such as the structure of multiplicative sub-groups and field extensions. In the case of $h=3,$ we deviate from most existing constructions of MR LRCs in that we do not use linearized constraints $(x,x^q,x^{q^2})$ or Vandermonde constraints $(x,x^2,x^3)$ and instead rely on Cauchy matrices~\cite{LN} to specify heavy parities.

Our construction of MR $(n,r=3,h=3,a=1,q)$-LRCs is technically disjoint from our other results. We observe that in this narrow case, MR LRCs are equivalent to subsets $A$ of the projective plane $\PF_q^2,$ where $A$ is partitioned in to triples $A = \sqcup_i\{a_i,b_i,c_i\}$ so that some three elements of $A$ are collinear if and only if they constitute one of the triples $\{a_i,b_i,c_i\}$ in the partition. Moreover, minimizing the field size of maximally recoverable local reconstruction codes is in fact equivalent to maximizing the cardinality of such sets $A.$ By considering all the $q+1$ lines through an arbitrary point of $A,$ it is easy to see that $|A|\le q+3.$ We construct sets $A$ with size $|A|\ge q^{1-o(1)}.$ For our construction we start with an elliptic curve $E$ over $\F_q$ such that the group of $\F_q$-rational points, $E(\F_q)$, is a cyclic group of size $\Omega(q)$. We observe that three points of $E(\F_q)$ are collinear if only and only if they sum to zero in the group. We then select a large AP-free set of points of $E(\F_q)$ using the classical construction of Behrend~\cite{Beh46} and complete these points to desired triples.

\subsection{Related work}\label{SubSec:Related}
The first family of codes with locality for applications in storage comes from~\cite{HCL,CHL}. These papers also introduced the concept of maximal recoverability in a certain restricted setting. The work of~\cite{GHSY} introduced a formal definition of local recovery and focused on codes that guarantee local recovery for a single failure. For this simple setting they were able to show that optimal codes must have a certain natural topology, e.g., codeword coordinates have to be arranged in groups where each group has a local parity. While~\cite{GHSY} focused on systematic codes that provide local recovery for information symbols,~\cite{Dimakis_0} considered codes that provide locality for all symbols and defined local reconstruction codes. In parallel works maximally recoverable  LRCs have been studied in~\cite{BHH,Blaum}. Construction of local reconstruction codes with optimal distance over fields of linear size has been given in~\cite{TB}. (Note that distance optimality is a much weaker property than maximal recoverability, e.g., when $a+h<r$ it only requires all patterns of size $a+h$ to be correctable, while MR property requires lots of very large patterns including some of size $(a+1)h$ to be correctable.)

Maximal recoverability can be defined with respect to more general topologies then just local reconstruction codes~\cite{GHJY}. The first lower bound for the field size of MR
codes in any topology was recently given in~\cite{GHKSWY}. This line of work was continued in~\cite{KLR} where nearly matching upper and lower bounds were obtained. The topology considered in~\cite{GHKSWY,KLR} is a grid-like topology, where codewords form a codimension one subspace of tensor product codes, i.e., codewords are
matrices, there is one heavy parity symbol, and each row / column constitutes a local group with one redundant symbol.

Finally, there are few other models of erasure correcting codes that provide efficient recovery in typical failure scenarios. These include regenerating codes~\cite{Dimakis_1,WTB,YB,GW} that optimize bandwidth consumed during repair rather than the number of coordinates (machines) accessed during repair; locally decodable codes~\cite{Y_now} that guarantee sub-linear time recovery of information coordinates even when a constant fraction of coordinates are erased; and SD codes~\cite{Blaum,BPSY} that correct a certain subset of failure patterns correctable by MR LRCs.

\subsection{Organization}\label{SubSec:Organization}
In Section~\ref{Sec:Prelim}, we setup our notation, give formal definitions of local reconstruction codes and maximal recoverability, and establish some basic facts about MR LRCs. In Section~\ref{Sec:LowerBound}, we present our main lower bound on the alphabet size. In Section~\ref{Sec:H2}, we introduce the determinantal identity and use it to give a construction of MR LRCs with two heavy parity symbols over fields of linear size. In Section~\ref{Sec:H3}, we get explicit MR codes over fields of cubic size. Finally, in Section~\ref{Sec:Elliptic}, we focus on the narrow case of codes with three heavy parities, one parity per local group, and local groups of size three. We introduce the machinery of elliptic curves and AP free sets and employ it to obtain maximally recoverable codes over fields of nearly linear size. We conclude by listing some open problems in Section~\ref{Sec:Open}. Appendix contains some missing proofs and proofs of the determinantal identities.

\section{Preliminaries}\label{Sec:Prelim}
We begin by summarizing few standard facts about erasure correcting codes~\cite{MS}.
\begin{itemize}
  \itemsep=-0.25ex
\item $[n,k,d]_q$ denotes a linear code (subspace) of dimension $k,$ codeword length $n,$ and Hamming distance $d$ over a field $\mathbb{F}_q.$ We often write $[n,k,d]$ or $[n,k]$ instead of $[n,k,d]_q$ when the left out parameters are not important.

\item An $[n,k,d]$ code is called Maximum Distance Separable (MDS) if $d=n-k+1.$

\item A linear $[n,k,d]_q$ code $C$ can be specified via its parity check matrix $H\in\mathbb{F}_q^{(n-k)\times n},$ where $C=\{x\in \mathbb{F}_q^n\mid H\cdot x=0\}.$ A code $C$ is MDS iff every $(n-k)\times (n-k)$ minor of $H$ is non-zero.

\item Let $C$ be an $[n,k]$ code with a parity check matrix $H\in\mathbb{F}^{(n-k)\times n}.$ Let $E$ be a subset of the coordinates of $C.$ If coordinates in $E$ are erased; then they can be recovered (corrected) iff the matrix $H$ restricted to coordinates in $E$ has full rank.
\end{itemize}
We proceed to formally define local reconstruction codes.
\begin{definition}\label{Def:LRC}
Let $r \mid n,$  $a<r,$ and $h$ be integers and $q$ be a prime power. Let $g=\frac{n}{r}.$ Assume $h\leq n-ag$ and let $k=n-ga-h.$ A linear $[n,k]$ code $C$ over a field $\mathbb{F}_q$ is an  $(n,r,h,a,q)$-LRC if for each $i\in [g],$ restricting $C$ to coordinates in $\{r(i-1)+1,\ldots,ri\},$ yields a maximum distance separable code with parameters $[r,r-a,a+1].$
\end{definition}

Let $[n]=\{1,\ldots,n\}.$ In what follows we refer to subsets $\{r(i-1)+1,\ldots,ri\}$ of the set of code coordinates $[n]$ as local groups. There are $g$ local groups and each such group has size $r.$ It is immediate from the Definition~\ref{Def:LRC} that every $(n,r,h,a,q)$-LRC admits a parity check matrix $H$ of the following form
\begin{equation}\label{fig:MRtopology}
H=
\left[
\begin{array}{c|c|c|c}
A_1 & 0 & \cdots & 0\\
\hline
0 &A_2 & \cdots & 0\\
\hline
\vdots & \vdots & \ddots & \vdots \\
\hline
0 & 0 & \cdots & A_g \\
\hline
B_1 & B_2 & \cdots & B_g \\
\end{array}
\right].
\end{equation}
Here $A_1,A_2,\cdots,A_g$ are $a\times r$ matrices over $\F_q$, $B_1,B_2,\cdots,B_g$ are $h\times r$ matrices over $\F_q.$ The rest of the matrix is filled with zeros. Every matrix $\{A_i\}_{i\in [g]}$ is a parity check matrix of an $[r,r-a,a+1]$ MDS code. The bottom $h$ rows of $H$ serve to increase the code co-dimension from $ag$ to $ag+h$. Conversely, every matrix $H$ as in~(\ref{fig:MRtopology}), where $\rk(H)=ag+h,$ and every $a\times a$ minor in each $\{A_i\}_{i\in [g]}$ is non-zero, defines an $(n,r,h,a,q)$-LRC. Note that  the bottom $h$ rows of the parity check matrix $H$ in~(\ref{fig:MRtopology}) can be chosen in any way and Definition~\ref{Def:LRC} does not impose any conditions on this. The minimum distance of a $(n,r,h,a,q)$-LRC is at most $a+h$.

\begin{definition}\label{Def:MRLRC}\footnote{Alternatively, one could define MR LRCs is as follows. Consider a matrix~(\ref{fig:MRtopology}). Each way of fixing non-zero entries in~(\ref{fig:MRtopology}) gives rise to (instantiates) a linear code. An instantiation is MR if it corrects all erasure patterns that are correctable for some other instantiation. It can be shown that under such definition and the minor technical assumption of $h\leq \frac{n}{r}\cdot (r-a)-\max\left\{\frac{n}{r},r-a\right\}$ local codes have to be MDS~\cite[Proposition 4]{GHKSWY} as required in Definition~\ref{Def:LRC}.}
Let $C$ be an arbitrary $(n,r,h,a,q)$-local reconstruction code. We say that $C$ is maximally recoverable if for any set $E\subseteq [n],$ $|E|=ga+h,$ where $E$ is obtained by selecting $a$ coordinates from each of $g$ local groups and then $h$ more coordinates arbitrarily; $E$ is correctable by the code $C.$
\end{definition}
The term maximally recoverable code is justified by the following observation (e.g.,~\cite{GHJY}): if an erasure pattern cannot be obtained via the process detailed in the Definition~\ref{Def:MRLRC}; then it cannot be corrected by any linear code whose parity check matrix has the shape~(\ref{fig:MRtopology}). Thus MR codes provide the strongest possible reliability guarantees given the locality constraints defining the shape of the parity check matrix.

Existence of MR LRCs can be established non-explicitly~\cite{GHJY} (i.e., by setting the non-zero entries in the matrix~(\ref{fig:MRtopology}) at random in a large finite field and then analyzing the properties of the resulting code). There are also multiple explicit constructions available~\cite{CK,GHJY,GYBS}. The key challenge in this line of work is to determine the minimal size of finite fields where such codes exist. In practice one is naturally mostly interested in fields of characteristic two.

\smallskip\noindent
\textbf{Notation:} We use $A\gtrsim B$ to denote $A=\Omega(B)$ and $A\lesssim B$ to denote $A=O(B)$. We use $A=O_\ell(B)$ and $A=\Omega_\ell(B)$ to denote that the hidden constants can depend on some parameter $\ell$ but independent of other parameters.

Given an $m\times n$ matrix $A$ and a subset $S\subset [m]$ of its rows and a subset $T\subset [n]$ of its columns, $A^{(S)}$ denotes the matrix formed by the rows of $A$ in $S$ and $A(T)$ denotes the matrix formed by the columns of $A$ in $T$.

\section{The lower bound}\label{Sec:LowerBound}

In this Section we prove Theorem~\ref{Th:lowerbound_mega} which gives a lower bound on the field size of maximally recoverable local reconstruction codes. We break up the proof of Theorem~\ref{Th:lowerbound_mega} into two cases based on $g\ge h$ and $g<h$ and prove the two cases in Corollary~\ref{Cor:LowerAsymptotic} and Proposition \ref{prop:LowerAsymptotic_gsmall} respectively. Though the underlying ideas in the lower bound for both the cases are very similar, the $g\ge h$ case is simpler and conveys all the main conceptual ideas. So we will prove this case first.
\subsection{Lower bound when $g\ge h$}
A code is MR if it corrects every erasure pattern that can be obtained by erasing $a$ symbols per local group, and then $h$ more. Note that if some local group carries at most $a$ erasures; then it can be immediately corrected using only the properties of the local MDS code. Thus we never need to consider erasure patterns spread across more than $h$ groups. Our lower bound does not use all the properties of MR LRCs, but only relies on code's ability to correct all patterns obtained by erasing $a+h$ elements in a single  group as well as all patterns obtained by erasing exactly $a+1$ coordinates in some $h$ local groups. Note that here we use the fact that the number of local groups $g$ is at least $h$.

The lower bound is obtained by turning a parity check matrix of an MR $(n,r,h,a,q)$-LRC into a large collection of points (of size $\approx nr^a$ when $a\le h-2$) in the projective space $\PF_q^{h-1},$ partitioned into $g$ equal parts $X_1,\ldots,X_{g}$, such that no hyperplane can intersect $h$ distinct sets in $\{X_j\}_{j\in [g]}.$ For example when $a=1$ and $h\ge 3$, the set $X_i$ is all the pairwise differences of columns of $B_i$ in (\ref{fig:MRtopology}) thought of as points in $\PF_q^{h-1}$ and so $|X_i|=\binom{r}{2}$.  In Lemma~\ref{lem:hyperplane_incidence}, we prove the size of such a collection can be at most $O(q)$ which implies the required lower bound. We will start by proving Lemma~\ref{lem:hyperplane_incidence}.

  \begin{lemma}\label{lem:hyperplane_incidence}
Let $X_1,\dots,X_g \subseteq \PF_q^d$ be mutually disjoint subsets each of size $t$ with $g\ge d+1$. If
\begin{equation}\label{Eqn:GeneralPos}
q<\left(\frac{g}{d}-1\right)t-4
\end{equation}
then there exists a hyperplane $H$ in $\PF_q^d$ which intersects $d+1$ distinct subsets among $X_1,\cdots,X_g.$
\end{lemma}
\begin{proof}
We will show that a random hyperplane will intersect $d+1$ distinct subsets among $X_1,\dots,X_g$ with positive probability if $q<\left(\frac{g}{d}-1\right)t-4$. Choose a uniformly random hyperplane $H$ in $\PF_q^d$. Fix some $i\in [g]$, we will first lower bound the probability that $H$ intersects $X_i$. Let the random variable $Z=|H\cap X_i|$. Since a hyperplane contains $|\PF_q^{d-1}|$ points, $$\E[Z]=\frac{|\PF_q^{d-1}|}{|\PF_q^d|}t.$$ We can also estimate the second moment as follows:
\begin{align*}
\E[Z^2]&=\E[Z]+\sum_{p,p'\in X_i, p\ne p'}\Pr[p,p'\in H]\\
&=\E[Z]+t(t-1)\frac{|\PF_q^{d-2}|}{|\PF_q^d|}
\end{align*}
where we used the fact that the number of hyperplanes containing two fixed distinct points is $|\PF_q^{d-2}|$. Note that $|\PF_q^d|=q^d+q^{d-1}+\dots+q+1=(q^{d+1}-1)/(q-1).$
Now we can lower bound $\Pr[Z>0]$ as:
\begin{align*}
\Pr[Z>0]&\ge \frac{\E[Z]^2}{\E[Z^2]}\\
&= \frac{\frac{(q^d-1)^2t^2}{(q^{d+1}-1)^2}}{(\frac{q^d-1)t}{(q^{d+1}-1)}+\frac{t(t-1)(q^{d-1}-1)}{(q^{d+1}-1)}}\\
& \ge \frac{(t^2/q^2)(1-1/q^d)^2}{t/q+t(t-1)/q^2}\\
& \ge \frac{t/q}{1+t/q} (1-1/q^d)^2.
\end{align*}
Since $X_1,\dots,X_g$ are mutually disjoint subsets of $\PF_q^d$ of size $t$, $gt\le |\PF_q^d| \le (d+1)q^d$. Therefore 
\begin{align*}
\Pr[H\cap X_i\ne \phi]&=\Pr[Z>0]\ge \frac{t}{t+q} \left(1-\frac{2}{q^d}\right)\ge \frac{t}{t+q}\left(1-\frac{2(d+1)}{gt}\right).
\end{align*}
By linearity of expectation, a random hyperplane $H$ intersects $\ge g\cdot \frac{t}{t+q}\left(1-\frac{2(d+1)}{gt}\right)$ sets among $X_1,\dots,X_g$ in expectation. Therefore if $ \frac{gt}{(q+t)} \left(1-\frac{2(d+1)}{gt}\right) >  d$, there exists a hyperplane which intersects $d+1$ distinct subsets among $X_1,\dots,X_g$. Rearranging this inequality, such a hyperplane exists whenever $q<\left(\frac{g}{d}-1\right)t-\frac{2(d+1)}{d}$.
 \end{proof}

We are now ready to prove the lower bound. We will first prove a lower bound under the assumption that $a+2 \le h$. Later in Proposition~\ref{prop:MainParikshit}, we generalize our argument to take care of the case when $h<a+2.$ 

\begin{proposition}\label{prop:MainOur}
When $a+2\le h \le n/r$, any maximally recoverable $(n,r,h,a,q)$-local reconstruction code must have
\begin{equation}\label{Eqn:MainLowerFormOur}
q \geq \left(\frac{n/r}{h-1}-1\right) \cdot \binom{r}{a+1}-4
\end{equation}
\end{proposition}
\begin{proof}
It might be helpful to the reader to think of the $a=1$ case through out the proof, as things get simpler. When $a=1$, wlog, one can assume that the entries of the matrices $A_i$ in (\ref{eqn:MRparitycheck}) (which will have only one row) are all 1's.

Consider an arbitrary maximally recoverable $(n,r,h,a,q)$-LRC $C$ with $g=\frac{n}{r}$ local groups. According to the discussion in Section~\ref{Sec:Prelim} the code $C$ admits a parity check matrix of the shape
\begin{equation}
\label{eqn:MRparitycheck}
\left[
\begin{array}{c|c|c|c}
A_1 & 0 & \cdots & 0\\
\hline
0 &A_2 & \cdots & 0\\
\hline
\vdots & \vdots & \ddots & \vdots \\
\hline
0 & 0 & \cdots & A_g \\
\hline
B_1 & B_2 & \cdots & B_g \\
\end{array}
\right].
\end{equation}
Here $A_1,A_2,\cdots,A_g$ are $a\times r$ matrices over $\F_q$, $B_1,B_2,\cdots,B_g$ are $h\times r$ matrices over $\F_q.$ The rest of the matrix is filled with zeros. Every $a\times a$ minor in each matrix $\{A_i\}_{i\in [g]}$ is non-zero. So for every subset $S\subset [r]$ of size $|S|=a+1$, $A_i(S)$ is an $a\times(a+1)$ matrix of full rank. Let $A_i(S)^\perp\in \F_q^{a+1}$ be a non-zero vector orthogonal to the row space of $A_i(S)$ i.e. $A _i(S)A_i(S)^\perp=0$. Note that $A_i(S)^\perp$ is unique upto scaling. For $i \in [g]$ and each subset $S\subseteq [r]$ of size $|S|=a+1,$ define $p_{i,S} \in \F_q^h$ as 
\footnote{When $a=1$, one can take $A_i(S)^\perp=\mattwoone{1}{-1}$ and so $p_{i,S}=B_i(j)-B_i(j')$ where $S=\{j,j'\}$; therefore $\{p_{i,S}: |S|=a+1\}$ is just the set of all pairwise differences of columns of $B_i$.}
\begin{align*}
  p_{i,S}
  = B_i(S)A_i(S)^\perp.
\end{align*}

The MR property implies that any subset of columns of the parity check matrix~(\ref{eqn:MRparitycheck}) which can be obtained by picking $a$ columns in each local group and $h$ arbitrary additional columns is full rank. We will use this property to make two claims about the vectors $\left\{p_{i,S}\right\}.$
\begin{claim}\label{claim:diff_groups}
For every distinct $\ell_1,\cdots, \ell_h\in [g]$ and subsets $S_1,\cdots,S_h\subseteq [r]$ of size $a+1$ each, the $h\times h$ matrix $\left[p_{\ell_1,S_1},\cdots, p_{\ell_h,S_h}\right]$ is full rank.
\end{claim}
\begin{proof}
Consider the following matrix equation:
\begin{equation*}
\left[
\begin{array}{c|c|c|c}
A_{\ell_1}(S_1) & 0 & \cdots & 0\\
\hline
0 &A_{\ell_2}(S_2) & \cdots & 0\\
\hline
\vdots & \vdots & \ddots & \vdots \\
\hline
0 & 0 & \cdots & A_{\ell_h}(S_h) \\
\hline
B_{\ell_1}(S_1) & B_{\ell_2}(S_2) & \cdots & B_{\ell_h}(S_h) \\
\end{array}
\right]
\left[
\begin{array}{c|c|c|c}
A_{\ell_1}(S_1)^\perp & 0 & \cdots & 0\\
\hline
0 &A_{\ell_2}(S_2)^\perp & \cdots & 0\\
\hline
\vdots & \vdots & \ddots & \vdots \\
\hline
0 & 0 & \cdots & A_{\ell_h}(S_h)^\perp \\
\end{array}
\right]
=
\left[
\begin{array}{c|c|c|c}
0 & 0 & \cdots & 0\\
\hline
0 &0 & \cdots & 0\\
\hline
\vdots & \vdots & \ddots & \vdots \\
\hline
0 & 0 & \cdots & 0 \\
\hline
p_{\ell_1,S_1} & p_{\ell_2,S_2} & \cdots & p_{\ell_h,S_h}\\
\end{array}
\right].
\end{equation*}
Let us denote the matrices which occur in the above equation as $M_1,M_2,M_3$ respectively so that the above equation becomes $M_1M_2=M_3$. By the MR property, when we erase the coordinates corresponding to $S_1,\cdots,S_h$ in groups $\ell_1,\cdots,\ell_h$ respectively, the resulting erasure pattern is correctable. This implies that $M_1$ has full rank. Also $M_2$ has full column rank because its columns are non-zero and have disjoint support. Therefore $M_3$ should have full rank which implies that $\left[p_{\ell_1,S_1},\cdots, p_{\ell_h,S_h}\right]$ is full rank.
\end{proof}
In particular the vectors $p_{i,S}$ are non-zero for every $i\in [g]$ and $S\in \binom{[r]}{a+1}$. We can also conclude that across different local groups, $p_{i,S}$ and $p_{j,T}$ are never multiples of each other when $i\ne j.$ In fact, we will now show that even in the same local group, $p_{i,S}$ and $p_{i,T}$ are not multiples of each other unless $S=T.$
\begin{claim}\label{claim:same_group}
For every $i\in [g],$ no two vectors in $\set{p_{i,S}: S\subseteq \binom{[r]}{a+1}}$ are multiples of each other.
\end{claim}
\begin{proof}
Suppose $p_{i,S}=\lambda\cdot p_{i,T}$ for some distinct sets $S,T\subset [r]$ of size $a+1$ each and some non-zero $\lambda\in \F_q.$ So,
\begin{align*}
\begin{bmatrix}
A_i(S)\\
B_i(S)\\
\end{bmatrix}
&\cdot
A_i(S)^\perp
-\lambda \cdot
\begin{bmatrix}
A_i(T)\\
B_i(T)\\
\end{bmatrix}
\cdot
A_i(T)^\perp\\
&=
\mattwoone{0}{p_{i,S}}-\lambda\cdot \mattwoone{0}{p_{i,T}}=0.
\end{align*}
Note that every coordinate of $A_i(S)^\perp$ is non-zero. If not, then it will imply a linear dependency between $a$ columns of $A_i(S)$ whereas we know that every $a\times a$ minor of $A_i(S)$ is non-zero. Thus we have a linear combination of the columns of $\mattwoone{A_i(S\cup T)}{B_i(S\cup T)}$ which is zero. Moreover the combination is non-trivial because there is some $j\in S\setminus T$ and the column $A_i(j)$ has a non-zero coefficient. However
\begin{equation}\label{Eqn:UnionSize}
|S\cup T|\le 2a+2\le a+h.
\end{equation}
By the MR property, any set of columns of the matrix $\mattwoone{A_i}{B_i}$ of size at most $a+h$ has to be full rank, as this set can be obtained by selecting (a subset of) $a$ and then $h$ more columns from the matrix~(\ref{eqn:MRparitycheck}). Thus we arrive at a contradiction that completes the proof of the claim.
\end{proof}
By Claim~\ref{claim:same_group} and the discussion above the claim, we can think of $\left\{p_{i,S}: i\in [g], S\in \binom{[r]}{a+1}\right\}$ as distinct points in $\PF_q^{h-1}.$ For brevity, from here on we assume that $p_{i,S}$ refers to the corresponding point in $\PF_q^{h-1}.$
Define sets $X_1,\cdots,X_g\subseteq \PF_q^{h-1}$ as $X_i=\left\{p_{i,S}: S \in \binom{[r]}{a+1}\right\},$ we have $|X_1|=|X_2|=\cdots=|X_g|=\binom{r}{a+1}$ and they are mutually disjoint. Also $g\ge h$ by the hypothesis. By Claim~\ref{claim:diff_groups}, there is no hyperplane in $\PF_q^{h-1}$ which contains $h$ points from distinct subsets of $X_1,\cdots, X_g.$ So applying Lemma~\ref{lem:hyperplane_incidence},
$$q \geq \left(\frac{g}{h-1}-1\right) \cdot \binom{r}{a+1}-4,$$
which concludes the proof.
\end{proof}
In the argument above we used vectors $\left\{p_{i,S}\right\},$ where $i$ varies across indices of $g$ local groups and $S$ varies across all $\binom{r}{a+1}$ subsets of $[r]$ of size $a+1.$ In the proof we relied on the condition $a+2 \leq h$ to ensure that the union of any two such sets $S$ has size at most $a+h.$

Parikshit Gopalan~\cite{Gopalan} has observed (and kindly allowed us to include his observation here) that we can generalize Proposition~\ref{prop:MainOur} to the case when $2\le h<a+2.$ To do this, in cases when $h<a+2$ we only consider sets $S$ that have size $a+1$ but are constrained to contain the set $\{1,2,\ldots,a+2-h\},$ as this ensures that pairwise unions still have size at most $a+h.$ Clearly, the total number of such sets is $\binom{r-a+h-2}{h-1}.$ The rest of the proof remains the same and yields the following
\begin{proposition}\label{prop:MainParikshit}
Assume $2\le h<a+2$ and $h\leq n/r;$ then any maximally recoverable $(n,r,h,a,q)$-local reconstruction code must have
\begin{equation}\label{Eqn:MainLowerFormOur2}
q \geq \left(\frac{n/r}{h-1}-1\right) \cdot \binom{r-a+h-2}{h-1} -4.
\end{equation}
\end{proposition}

The following corollary follows immediately from Propositions~\ref{prop:MainOur} and ~\ref{prop:MainParikshit} and presents the asymptotic form of our field size lower bound when $g\ge h$.
\begin{corollary}\label{Cor:LowerAsymptotic}
Suppose that $a$ and $h\ge 2$ are arbitrary constants, but $r$ may grow with $n.$ Further suppose that $h\leq n/r.$ In every maximally recoverable $(n,r,h,a,q)$-LRC, we have:
\begin{equation}\label{Eqn:MainLowerRepeat}
q\geq \Omega_{a,h}\left(n\cdot r^{\min\{a,h-2\}}\right).
\end{equation}
\end{corollary}

\subsection{Lower bound when $g\le h$}
In this case, we cannot distribute the $h$ additional erasures among $h$ different local groups. Instead we will look at erasure patterns where either all the extra $h$ erasures occur in the same group or they are spread equally ($\ceil{h/g}$ or $\floor{h/g}$) in the $g$ local groups. The sets $X_1,\dots,X_g$ will now be a collection of subspaces of dimension roughly $h/g$ such that no $(h-1)$-dimensional subspace can contain a subspace each from all of $X_1,\dots,X_g$. To obtain the lower bound, we show that if $q$ is too small, a random $(h-1)$-dimensional subspace will contain a subspace from each of $X_1,\dots,X_g$ with high probability. The argument is more involved than in the $g \ge h$ case, because the subspaces inside each $X_i$ can intersect non-trivially and the analysis has to account for this carefully. We obtain the following lower bound, the proof of which appears in Section~\ref{Sec:proofof_lowerasymptotic_gsmall}.

\begin{proposition}\label{prop:LowerAsymptotic_gsmall}
Suppose that $a,g,h$ are fixed constants such that $2\le g \le h$. In every maximally recoverable $(n,r,h,a,q)$-LRC with $g$ local groups each of size $r=n/g$, we have:
\begin{equation}
q\geq \Omega_{a,h,g}\left(n^{1+\alpha}\right) \text{ where } \alpha=\frac{\min\{a,h-2\ceil{h/g}\}}{\ceil{h/g}}.
\end{equation}
\end{proposition}

\section{Maximally recoverable LRCs with $h=2$}\label{Sec:H2}
In this section we present our construction of maximally recoverable local reconstruction codes with two heavy parity symbols. Our construction relies on a determinantal identity (Lemma~\ref{lem:blockdet}) and properties of $\F_q^*$, the multiplicative group of the field $\mathbb{F}_q.$ The following identity conveniently reduces the $(ah+h)\times (ah+h)$ determinants that arise during our analysis to $h\times h$ determinants which are much easier to calculate. We will prove Lemma~\ref{lem:blockdet} in Section~\ref{Sec:determinantal}.

\begin{lemma}\label{lem:blockdet} Let $C_1,\cdots,C_h$ be $a\times (a+1)$ dimensional matrices and $D_1,\cdots,D_h$ be $h \times (a+1)$ dimensional matrices over a field and let $D_i^{(j)}$ be the $j^{th}$ row of $D_i$. Then,
\begin{align*}
&\det \left[
\begin{array}{c|c|c|c}
C_1 & 0 & \cdots & 0\\
\hline
0 &C_2 & \cdots & 0\\
\hline
\vdots& \vdots & \ddots & \vdots \\
\hline
0 & 0 & \cdots & C_h \\
\hline
D_1 & D_2 & \cdots & D_h \\
\end{array}
\right]= (-1)^{\frac{ah(h-1)}{2}}
\det\left[
\begin{matrix}
\det\mattwoone{C_1}{D_1^{(1)}}& \cdots & \det\mattwoone{C_h}{D_h^{(1)}}\\
\vdots & \ddots & \vdots \\
\det\mattwoone{C_1}{D_1^{(h)}}& \cdots & \det\mattwoone{C_h}{D_h^{(h)}}\\
\end{matrix}
\right].
\end{align*}
\end{lemma}

\begin{lemma}\label{lem:H2Preciese}
Let $r \mid n,$  $a<r$ be integers. Let $g=\frac{n}{r}.$ Assume that $n-ga-2$ is positive. Suppose $q$ is a prime power such that there exists a subgroup of $\F_q^*$ of size at least $r$ and with at least $n/r$ cosets; then there exists an explicit maximally recoverable $(n,r,h=2,a,q)$-local reconstruction code.
\end{lemma}
\begin{proof}
Let $G\subset \F_q^*$ be the multiplicative subgroup from the statement of the Lemma. Let $\alpha_1,\alpha_2,\cdots,\alpha_r \in G$ be distinct elements from $G$ and let $\lambda_1,\lambda_2,\cdots,\lambda_g \in \F_q^*$ be elements from distinct cosets of $G$. We specify our code via a parity check matrix of the form~(\ref{fig:MRtopology}). For $i\in [g],$ we choose matrices $\{A_i\}$ and $\{B_i\}$ as:
\begin{equation*}
A_i=
\left[
\begin{matrix}
\alpha_1 & \alpha_2 & \cdots & \alpha_r\\
\alpha_1^2 & \alpha_2^2 & \cdots & \alpha_r^2\\
\vdots&\vdots &\ddots & \vdots\\
\alpha_1^a & \alpha_2^a & \cdots & \alpha_r^a\\
\end{matrix}
\right];\quad B_i=
\left[
\begin{matrix}
\lambda_i & \lambda_i & \cdots & \lambda_i\\
\alpha_1^{a+1} & \alpha_2^{a+1} & \cdots & \alpha_r^{a+1}\\
\end{matrix}
\right].
\end{equation*}
Suppose that we have $a$ erasures per local group and two more. We can easily correct the coordinates corresponding to local groups which have at most $a$ erasures in them. This is because every matrix $A_i$ is a Vandermonde matrix and all its $a \times a$ minors are non-zero. Now we are left with two cases:

\medskip \noindent
\textbf{Case 1:} Both the extra erasures occurred in the same local group. Say, the $i^{th}$ local group. In this case, we can correct the erased coordinates because any $(a+2)\times (a+2)$ minor of $\left[\begin{matrix}
A_i \\ B_i\\
  \end{matrix}\right]$ (which is a Vandermonde matrix after scaling and permuting rows) is non-zero.

\medskip\noindent
\textbf{Case 2:} The two extra erasures occur in different groups say groups $\ell$ and $\ell'$, so we are left with two groups with $a+1$ erasures in each. Let $S$ be the columns erased in group $\ell$ and let $S'$ be the columns erased in group $\ell'$. We want to argue that the following $(2a+2)\times (2a+2)$ submatrix is full rank:
\begin{equation}
M=\left[
\begin{array}{c|c}
A_\ell(S)&0\\
\hline
0&A_{\ell'}(S')\\
\hline
B_\ell(S)& B_{\ell'}(S')\\
\end{array}
\right].
\end{equation}
Let $S=\set{\gamma_1,\gamma_2,\cdots,\gamma_{a+1}}$ and $S'=\set{\gamma'_1,\gamma'_2,\cdots,\gamma'_{a+1}}$, then
by Lemma~\ref{lem:blockdet},
\begin{align*}
\det(M)=0 &\iff
\det \left[
\begin{array}{cc}
\det\mattwoone{A_\ell(S)}{B_\ell(S)^{(1)}}&\det\mattwoone{A_{\ell'}(S')}{B_{\ell'}(S')^{(1)}}\\
\det\mattwoone{A_\ell(S)}{B_\ell(S)^{(2)}}&\det\mattwoone{A_{\ell'}(S')}{B_{\ell'}(S')^{(2)}}\\
\end{array}
\right]=0\\
&\iff
\det\left[
\begin{array}{cc}
\det\left(\begin{matrix}
\gamma_1 & \cdots & \gamma_{a+1}\\
\gamma_1^2 & \cdots & \gamma_{a+1}^2\\
\vdots& \ddots & \vdots \\
\gamma_1^a & \cdots & \gamma_{a+1}^a\\
\lambda_\ell & \cdots & \lambda_\ell\\
\end{matrix}\right)
&
\det\left(\begin{matrix}
\gamma'_1 & \cdots & \gamma'_{a+1}\\
(\gamma'_1)^2 & \cdots & (\gamma'_{a+1})^2\\
\vdots& \ddots & \vdots \\
(\gamma'_1)^a & \cdots & (\gamma'_{a+1})^a\\
\lambda_{\ell'} & \cdots & \lambda_{\ell'}\\
\end{matrix} \right)\\
\det\left(\begin{matrix}
\gamma_1 & \cdots & \gamma_{a+1}\\
\gamma_1^2 & \cdots & \gamma_{a+1}^2\\
\vdots& \ddots & \vdots \\
\gamma_1^a & \cdots & \gamma_{a+1}^a\\
\gamma_1^{a+1} & \cdots & \gamma_{a+1}^{a+1}\\
\end{matrix} \right)
&
\det\left(\begin{matrix}
\gamma_1' & \cdots & \gamma'_{a+1}\\
{\gamma_1'}^2 & \cdots & (\gamma'_{a+1})^2\\
\vdots& \ddots & \vdots \\
{\gamma_1'}^a & \cdots & (\gamma'_{a+1})^a\\
{\gamma_1'}^{a+1} & \cdots & (\gamma'_{a+1})^{a+1}\\
\end{matrix}\right)\\
\end{array}
\right]=0\\
&\iff \det\left[
\begin{matrix}
\lambda_\ell & \lambda_{\ell'}\\
\prod_{i\in [a+1]} \gamma_i & \prod_{i\in [a+1]} \gamma'_i\\
\end{matrix}
\right]=0
 \end{align*}
 where we factored out the (non-zero) Vandermonde determinant from each column.
Since $\gamma_i, \gamma'_i \in G$ and $\lambda_\ell, \lambda_{\ell'}$ are in different cosets of $G$, the last determinant is not zero.
\end{proof}

In Lemma~\ref{lem:H2Preciese}, given $n$ and $r$ such that $r\mid n$, we want to find a small field $\F_q$ such that $\F_q^*$ contains a subgroup of size at least $r$ and with at least $n/r$ cosets. For example, if $n+1$ is a prime power, then we can take $q=n+1$. The following lemma shows that one can always find such a field of size $q=O(n)$. We prove it in Section~\ref{Sec:subgroup_lemma}.

\begin{lemma}\label{lem:subgroup}
Let $r,n$ be some positive integers with $r\le n$. Then there exists a finite field $\F_q$ with $q=O(n)$ such that the multiplicative group $\F_q^*$ contains a subgroup of size at least $r$ and with at least $n/r$ cosets. If additionally we require that the field has characteristic two, then such a field exists with $q=n\cdot \exp(O(\sqrt{\log n})).$
\end{lemma}

\noindent Combining Lemma~\ref{lem:subgroup} with Lemma~\ref{lem:H2Preciese} gives the following theorem.

\begin{theorem}\label{Th:H2Main}
Let $r \mid n,$  $a<r$ be integers. Let $g=\frac{n}{r}.$ Assume that $n-ga-2$ is positive. Then there exists an explicit maximally recoverable $(n,r,h=2,a,q)$-local reconstruction code with $q=O(n).$ If we require the field to be of characteristic 2, such a code exists with $q \le n\cdot\exp(O(\sqrt{\log n})).$
\end{theorem}

\section{Maximally recoverable LRCs with $h=3$}\label{Sec:H3}
In this section, we present our construction of maximally recoverable local reconstruction codes with three heavy parity symbols. Our construction extends the ideas in the construction of Section~\ref{Sec:H2} using field extensions. In addition to the determinantal identity~\ref{lem:blockdet}, we will need the following identity which follows immediately from Lemma~\ref{lem:blockdet_gen}.

\begin{lemma}\label{lem:blockdet21} Let $C_1$ be an $a\times (a+1)$ matrix, $C_2$ be an $a\times (a+2)$ matrix, $D_1$ be a $3 \times (a+1)$ matrix and $D_2$ be a $3\times (a+2)$ matrix and let  $D_i^{(j)}$ be the $j^{th}$ row of $D_i$. Then,
\begin{align*}
\det \left[
\begin{array}{c|c}
C_1 & 0 \\
\hline
0 &C_2\\
\hline
D_1 & D_2 \\
\end{array}
\right]
=0
\iff 
\det\mattwoone{C_1}{D_1^{(1)}} \cdot \det\matthreeone{C_2}{D_2^{(2)}}{D_2^{(3)}}
- \det\mattwoone{C_1}{D_1^{(2)}} \cdot \det\matthreeone{C_2}{D_2^{(1)}}{D_2^{(3)}} + \det\mattwoone{C_1}{D_1^{(3)}}\cdot \det\matthreeone{C_2}{D_2^{(1)}}{D_2^{(2)}}
=0.
\end{align*}
\end{lemma}
Our construction is based on Cauchy matrices, so we will also need the the following lemma about the determinants of such matrices.
\begin{lemma}(\cite{LN})\label{lem:cauchy}
 Let $\alpha_1,\cdots,\alpha_m, \beta_1,\cdots, \beta_m\in \F_q$ be all distinct; then
 \begin{align*}
 \det
 &\begin{bmatrix}
 \frac{1}{\alpha_1-\beta_1} &  \frac{1}{\alpha_2-\beta_1}& \cdots & \frac{1}{\alpha_m-\beta_1}\\
  \frac{1}{\alpha_1-\beta_2} &  \frac{1}{\alpha_2-\beta_2}& \cdots & \frac{1}{\alpha_m-\beta_2}\\
\vdots & \vdots&\ddots & \vdots\\
\frac{1}{\alpha_1-\beta_m} &  \frac{1}{\alpha_2-\beta_m}& \cdots & \frac{1}{\alpha_m-\beta_m}\\
 \end{bmatrix}=
 \frac{\prod_{i>j} (\alpha_i-\alpha_j)(\beta_j-\beta_i)}{\prod_{i,j}(\alpha_i-\beta_j)}.
 \end{align*}
 \end{lemma}
 Matrices of the above form are called Cauchy matrices. Every minor of a Cauchy matrix is non-zero because square submatrices of a Cauchy matrix are also Cauchy matrices. We are now ready to present the construction for three global parities.
\begin{lemma}\label{lem:H3Precise}
Let $r \mid n,$  $a<r$ be integers. Let $g=\frac{n}{r}$. Assume that $n-ga-3$ is positive. Suppose $q_0\ge 2r+3$ is a prime power such that there exists a subgroup of $\F_{q_0}^*$ of size at least $r+2$ and with at least $n/r$ cosets. Then there exists an explicit maximally recoverable $(n,r,h=3,a,q=q_0^3)$-local reconstruction code.
\end{lemma}
\begin{proof}
Let $G\subset \F_{q_0}^*$ be the multiplicative subgroup from the statement of the theorem. Choose distinct $\beta_{a+1}, \beta_{a+2}, \beta_{a+3}\in \F_{q_0}$ and let $$\Omega=\left\{\alpha\in \F_{q_0}: \frac{\alpha-\beta_{a+2}}{\alpha-\beta_{a+3}}\in G\right\}.$$
Clearly $|\Omega|= |G|-1\ge r+1$, so we can choose distinct $\alpha_1,\cdots,\alpha_r\in \Omega\setminus \set{\beta_{a+1}}$. Finally, since $q_0\ge 2r+3 \ge r+a+3$, we can choose distinct $\beta_1,\cdots,\beta_a \in \F_{q_0}\setminus \set{\alpha_1,\cdots,\alpha_r, \beta_{a+1},\beta_{a+2},\beta_{a+3}}$. Let $\mu_1,\cdots,\mu_g \in \F_{q_0}$ be elements from distinct cosets of $G$.

Now let $\F_q$ be a degree 3 extension of $\F_{q_0}$, so we have $q=q_0^3$. As $\F_q$ is a 3-dimensional vector space over $\F_{q_0}$, choose a basis $v_0,v_1,v_2\in \F_q$ for this space and choose distinct $\gamma_1,\cdots,\gamma_g\in \F_{q_0}$. Define $\lambda_i=v_0+\gamma_i v_1+\gamma_i^2 v_2$. Then any  three of the elements $\lambda_1,\cdots, \lambda_g\in \F_q$ are linearly independent over $\F_{q_0}$; we call this property $3$-wise independence over $\F_{q_0}$. Define the matrices $A_i$ and $B_i$ as follows:
\begin{align*}
A_i=
\left[
\begin{matrix}
\frac{1}{\alpha_1-\beta_1} & \cdots & \frac{1}{\alpha_r-\beta_1}\\
\vdots &\ddots & \vdots\\
\frac{1}{\alpha_1-\beta_a} & \cdots & \frac{1}{\alpha_r-\beta_a}\\
\end{matrix}
\right];\quad B_i=
\left[
\begin{matrix}
\frac{\lambda_i}{\alpha_1-\beta_{a+1}} & \cdots & \frac{\lambda_i}{\alpha_r-\beta_{a+1}}\\
\frac{\mu_i}{\alpha_1-\beta_{a+2}} & \cdots & \frac{\mu_i}{\alpha_r-\beta_{a+2}}\\
\frac{1}{\alpha_1-\beta_{a+3}} & \cdots & \frac{1}{\alpha_r-\beta_{a+3}}\\
\end{matrix}
\right].
\end{align*}
Now we will show that the above construction satisfies the MR property. We have $a$ erasures per local group and $3$ more. We can easily correct groups with only $a$ erasures because $A_i$ are Cauchy matrices where every $a\times a$ minor is non-zero. So we only need to worry about local groups with more than $a$ erasures. There are three cases.

\smallskip\noindent
\textbf{Case 1:} All three extra erasures in the same group.\\
Say we have $a+3$ erasures in local group $i$, then we can correct these errors because the matrix $\mattwoone{A_i}{B_i}$ is a Cauchy matrix (except for some scaling factors in the rows), and therefore each of its  $(a+3)\times (a+3)$ minors is non-zero by Lemma~\ref{lem:cauchy}.

\smallskip \noindent
\textbf{Case 2:} The three extra erasures are distributed across two groups.\\
Suppose the extra erasures occur in groups $\ell, \ell'$ with $(a+1)$ erasures in group $\ell$ corresponding to a subset $S\subseteq [r]$ of its columns and $(a+2)$ erasures in group $\ell'$ corresponding to a subset $S'\subseteq [r]$ of its columns. To correct these erasures we need to show the following matrix is full rank:
\begin{equation}
\label{eq:H3-2+1}
\left[
\begin{array}{c|c}
A_\ell(S) & 0 \\
\hline
0 &A_{\ell'}(S')\\
\hline
B_\ell(S) & B_{\ell'}(S') \\
\end{array}
\right]  \ .
\end{equation}
By Lemma~\ref{lem:blockdet21}, the above matrix fails to be full rank iff
\begin{align*}
\det\mattwoone{A_\ell(S)}{B_\ell(S)^{(1)}} \cdot \det\matthreeone{A_{\ell'}(S')}{B_{\ell'}(S')^{(2)}}{B_{\ell'}(S')^{(3)}}
- \det\mattwoone{A_\ell(S)}{B_\ell(S)^{(2)}} \cdot \det\matthreeone{A_{\ell'}(S')}{B_{\ell'}(S')^{(1)}}{B_{\ell'}(S')^{(3)}} 
+ \det\mattwoone{A_\ell(S)}{B_\ell(S)^{(3)}}\cdot \det\matthreeone{A_{\ell'}(S')}{B_{\ell'}(S')^{(1)}}{B_{\ell'}(S')^{(2)}}
=0.
\end{align*}
The above determinant is a $\F_q$-linear combination of $\lambda_\ell$ and $\lambda_{\ell'}$ and the coefficient of $\lambda_\ell$, which arises from the first term, is non-zero because $\mattwoone{A_\ell}{B_\ell}$ and $\mattwoone{A_{\ell'}}{B_{\ell'}}$ are Cauchy matrices. By $3$-wise independence of $\lambda$'s, this linear combination cannot be zero, and therefore the matrix \eqref{eq:H3-2+1} has full rank.

\smallskip\noindent
\textbf{Case 3:} The three extra erasures occur in distinct groups.\\
Suppose the three extra erasures occur in groups $\ell_1,\ell_2,\ell_3\in [g]$ and let $S_1,S_2,S_3\subseteq [r]$ be sets of size $a+1$ corresponding to the erasures in the groups $\ell_1,\ell_2,\ell_3$ respectively. To correct these erasures we need to show the following matrix is full rank:
\begin{align*}
\left[
\begin{array}{c|c|c}
A_{\ell_1}(S_1) & 0  & 0\\
\hline
0 &A_{\ell_2}(S_2)  & 0\\
\hline
0 & 0 & A_{\ell_3}(S_3) \\
\hline
B_{\ell_1}(S_1) & B_{\ell_2}(S_2) & B_{\ell_3}(S_3) \\
\end{array}
\right]
\end{align*}
By Lemma~\ref{lem:blockdet}, if the above matrix is not full rank then
\begin{align*}
\det\begin{bmatrix}
\det\mattwoone{A_{\ell_1}(S_1)}{B^{(1)}_{\ell_1}(S_1)} & \det\mattwoone{A_{\ell_2}(S_2)}{B^{(1)}_{\ell_2}(S_2)} & \det\mattwoone{A_{\ell_3}(S_3)}{B^{(1)}_{\ell_3}(S_3)}\\
\det\mattwoone{A_{\ell_1}(S_1)}{B^{(2)}_{\ell_1}(S_1)} & \det\mattwoone{A_{\ell_2}(S_2)}{B^{(2)}_{\ell_2}(S_2)} & \det\mattwoone{A_{\ell_3}(S_3)}{B^{(2)}_{\ell_3}(S_3)}\\
\det\mattwoone{A_{\ell_1}(S_1)}{B^{(3)}_{\ell_1}(S_1)} & \det\mattwoone{A_{\ell_2}(S_2)}{B^{(3)}_{\ell_2}(S_2)} & \det\mattwoone{A_{\ell_3}(S_3)}{B^{(3)}_{\ell_3}(S_3)}\\
\end{bmatrix}
=0.
\end{align*}
For $k\in \{1,2,3\}$, let $c_k=\prod_{i>j, i,j\in S_k}(\alpha_i-\alpha_j), d=\prod_{i>j, i,j\in [a]}(\beta_j-\beta_i), e_k=\prod_{i\in S_k, j\in [a]}(\alpha_i-\beta_j).$ By Lemma~\ref{lem:cauchy}, we can write down explicit expressions for the entries in the above determinant to get:
\begin{align*}
\det\begin{bmatrix}
\lambda_{\ell_1}\frac{c_1d \prod_{i\in[a]}(\beta_i-\beta_{a+1})}{e_1\prod_{i\in S_1}(\alpha_i-\beta_{a+1})} &
\lambda_{\ell_2}\frac{c_2d \prod_{i\in[a]}(\beta_i-\beta_{a+1})}{e_2\prod_{i\in S_2}(\alpha_i-\beta_{a+1})} &
\lambda_{\ell_3}\frac{c_3d \prod_{i\in[a]}(\beta_i-\beta_{a+1})}{e_3\prod_{i\in S_3}(\alpha_i-\beta_{a+1})} &
\\
\vspace{0.01cm}
\\
\mu_{\ell_1}\frac{c_1d \prod_{i\in[a]}(\beta_i-\beta_{a+2})}{e_1\prod_{i\in S_1}(\alpha_i-\beta_{a+2})} &
\mu_{\ell_2}\frac{c_2d \prod_{i\in[a]}(\beta_i-\beta_{a+2})}{e_2\prod_{i\in S_2}(\alpha_i-\beta_{a+2})} &
\mu_{\ell_3}\frac{c_3d \prod_{i\in[a]}(\beta_i-\beta_{a+2})}{e_3\prod_{i\in S_3}(\alpha_i-\beta_{a+2})} &
\\
\vspace{0.01cm}
\\
\frac{c_1d \prod_{i\in[a]}(\beta_i-\beta_{a+3})}{e_1\prod_{i\in S_1}(\alpha_i-\beta_{a+3})} &
\frac{c_2d \prod_{i\in[a]}(\beta_i-\beta_{a+3})}{e_2\prod_{i\in S_2}(\alpha_i-\beta_{a+3})} &
\frac{c_3d \prod_{i\in[a]}(\beta_i-\beta_{a+3})}{e_3\prod_{i\in S_3}(\alpha_i-\beta_{a+3})} &
\\
\end{bmatrix}
=0.
\end{align*}
We can scale rows and columns to conclude that
\begin{align*}
\det\begin{bmatrix}
\lambda_{\ell_1}\prod_{i\in S_1}\left(\frac{\alpha_i-\beta_{a+3}}{\alpha_i-\beta_{a+1}}\right) &
\lambda_{\ell_2}\prod_{i\in S_2}\left(\frac{\alpha_i-\beta_{a+3}}{\alpha_i-\beta_{a+1}}\right) &
\lambda_{\ell_3}\prod_{i\in S_3}\left(\frac{\alpha_i-\beta_{a+3}}{\alpha_i-\beta_{a+1}}\right)
\\
\mu_{\ell_1}\prod_{i\in S_1}\left(\frac{\alpha_i-\beta_{a+3}}{\alpha_i-\beta_{a+2}}\right) &
\mu_{\ell_2}\prod_{i\in S_2}\left(\frac{\alpha_i-\beta_{a+3}}{\alpha_i-\beta_{a+2}}\right) &
\mu_{\ell_3}\prod_{i\in S_3}\left(\frac{\alpha_i-\beta_{a+3}}{\alpha_i-\beta_{a+2}}\right)
\\
1 & 1 & 1
\end{bmatrix}
=0.
\end{align*}
By the choice of $\alpha$'s, $\prod_{i\in S_j}\left(\frac{\alpha_i-\beta_{a+3}}{\alpha_i-\beta_{a+2}}\right)\in G$ for $j=1,2,3$. By writing the Laplace expansion of the determinant over the first row, the above determinant is a linear combination in $\lambda_{\ell_1}, \lambda_{\ell_2},\lambda_{\ell_3}$ with coefficients from $\F_{q_0}$. The coefficients of $\lambda$'s in this linear combination are non-zero because $\mu_{\ell_1},\mu_{\ell_2},\mu_{\ell_3}$ belong to distinct cosets of $G$ in $\F_{q_0}^*$. Because $\lambda$'s are $3$-wise independent over $\F_{q_0}$, we get a contradiction.
\end{proof}
\noindent
Combining Lemma~\ref{lem:H3Precise} with Lemma~\ref{lem:subgroup} gives the following theorem.
\begin{theorem}\label{Th:H3Main}
Let $r \mid n,$  $a<r$ be integers. Let $g=\frac{n}{r}\ge 2.$ Assume that $n-ga-3$ is positive. Then there exists an explicit maximally recoverable $(n,r,h=3,a,q)$-local reconstruction code with $q=O(n^3).$ If we require the field to be of characteristic 2, such a code exists with $q=n^3\cdot\exp(O(\sqrt{\log n})).$
\end{theorem}

\section{Maximally recoverable LRCs from elliptic curves}\label{Sec:Elliptic}
Our construction of MR $(n,r=3,h=3,a=1,q)$-LRCs is technically
disjoint from our results in the previous sections. We observe that in
this narrow case, maximally recoverable LRCs are equivalent to
families of \emph{matching collinear triples} in the projective plane
$\PF_q^2,$ i.e., sets of points partitioned into collinear triples,
where no three points other than those forming a triple are
collinear. In Section~\ref{SubSec:LRCsFromTriples} we state the
quantitative parameters of such a family $A$ that we can obtain and
translate those to parameters of an MR LRC. The goal of
Section~\ref{Sec:matching_collinear} is to construct the family $A$ using elliptic curves and 3-AP free sets. In Section~\ref{SubSec:EliipticAndAP} we develop the necessary
machinery of elliptic curves, and in Section~\ref{SubSec:CollinearFamily} we carry out
the construction.

\subsection{LRCs from matching collinear triples}\label{SubSec:LRCsFromTriples}
We will reduce the problem of constructing maximally recoverable codes for $h=3,r=3,a=1$ to the problem of constructing matching collinear triples in $\PF_q^2$ which we define below.
\begin{definition}
We say that $A\subset \PF_q^2$ has \textit{matching collinear triples} if $A$ can be partitioned into triples, $A=\sqcup_{i=1}^m \{a_i,b_i,c_i\}$, such that the only collinear triples in $A$ are $\{a_i,b_i,c_i\}$ for  $i\in [m]$.
\end{definition}

What is the largest subset $A\subset \PF_q^2$ with matching collinear triples? If we consider all the $q+1$ lines through some fixed point of $A$, at most one line can contain two other points of $A$. All other lines can contain at most one other point of $A$. So $|A|\le q+3$. The following lemma shows that we can construct a set $A$ with size $|A|\ge q^{1-o(1)}$. It is an interesting open question if we can get $|A|\geq  \Omega(q)$.

\begin{lemma}\label{lem:matching_construction}
For any prime power $q$, there is an explicit set $A\subset \PF_q^2$ with matching collinear triples of size $|A|\ge q\cdot \exp(-C\sqrt{\log q})$ where $C>0$ is some absolute constant.
\end{lemma}
We will prove Lemma~\ref{lem:matching_construction} in Section~\ref{SubSec:CollinearFamily}.
\begin{lemma}\label{lem:reduction_matching_triples}
Assume $g\geq 2.$ There exists a subset $S\subset \PF_q^2$ that has $g$ matching collinear triples \emph{if and only if} there exists a maximally recoverable $(3g,r=3,h=3,a=1,q)$-local reconstruction code.
\end{lemma}
\begin{proof}
We first show how to obtain codes from families of collinear triples. Let $S=\cup_{i=1}^g \{a_i,b_i,c_i\}$ be such that the only collinear triples in $S$ are $\{a_i,b_i,c_i\}$ for $i\in [g]$. From now, we will think of elements of $S$ as vectors in $\F_q^3$ such that every triple of points except for the triples $\set{a_i,b_i,c_i}$ are linearly independent. We can scale each vector with non-zero elements in $\F_q$ such that $a_i+b_i+c_i=0$ in $\F_q^3$ for every $i\in [g].$ For $i\in [g],$ define blocks $A_i$ and $B_i$ of the parity check matrix~(\ref{fig:MRtopology}) as:
$$A_i=\begin{bmatrix}
1 & 1& 1\\
\end{bmatrix};\ B_i=\begin{bmatrix}
0&-b_i&c_i\\
\end{bmatrix}.$$
We need to correct 1 erasure per group and any 3 extra erasures. We can correct groups with a single erasure because $A_i$ is a simple parity check constraint on all the coordinates of the group. We now have to correct groups with more than one erasure, there are two cases:

\smallskip\noindent
\textbf{Case 1:} The three extra erasures are in two groups.\\
Suppose the two groups are $i,j$ and in group $i$ all the coordinates are erased and in group $j$ the second and third coordinates are erased (the other two cases are similar). To correct these erasures, we have to argue that the following matrix is full rank:
$$
\left[
\begin{array}{ccc|cc}
1& 1& 1 & 0 & 0\\
\hline
0 & 0&0& 1 & 1\\
\hline
0& -b_i &c_i & -b_j& c_j\\
\end{array}
\right]
$$
Subtract the first column in each group from the rest, it is equivalent to the following matrix being full rank:
\begin{align*}
\left[
\begin{array}{ccc|cc}
1& 0& 0 & 0 & 0\\
\hline
0 & 0&0& 1 & 0\\
\hline
0& -b_i &c_i & -b_j& c_j+b_j\\
\end{array}
\right]=
\left[
\begin{array}{ccc|cc}
1& 0& 0 & 0 & 0\\
\hline
0 & 0&0& 1 & 0\\
\hline
0& -b_i &c_i & -b_j& a_j\\
\end{array}
\right]
\end{align*}
which is true because $b_i,c_i, a_j$ are linearly independent.

\smallskip\noindent
\textbf{Case 2:} The three extra erasures are in distinct groups.\\
Suppose the three groups are $i,j,k$ and in each group the second and third columns are erased (the other cases are similar). To correct these erasures, we have to argue that the following matrix is full rank:
$$
\left[
\begin{array}{cc|cc|cc}
1& 1 & 0 & 0& 0&0\\
\hline
0 & 0&1& 1 & 0&0\\
\hline
0 & 0&0& 0 & 1&1\\
\hline
-b_i &c_i & -b_j& c_j&  -b_k& c_k\\
\end{array}
\right]
$$
Subtract the first column in each group from the rest, it is equivalent to the following matrix being full rank:
\begin{align*}
\left[
\begin{array}{cc|cc|cc}
1& 0 & 0 & 0& 0&0\\
\hline
0 & 0&1& 0 & 0&0\\
\hline
0 & 0&0& 0 & 1&0\\
\hline
-b_i &c_i+b_i & -b_j& c_j+b_j&  -b_k& c_k+b_k\\
\end{array}
\right]=
\left[
\begin{array}{cc|cc|cc}
1& 0 & 0 & 0& 0&0\\
\hline
0 & 0&1& 0 & 0&0\\
\hline
0 & 0&0& 0 & 1&0\\
\hline
-b_i & -a_i & -b_j& -a_j&  -b_k& -a_k\\
\end{array}
\right]
\end{align*}
which is true because $a_i,a_j, a_k$ are linearly independent.

\noindent
\textbf{Reverse connection.}
We now proceed to show how to obtain a set with matching collinear triples from codes. Given a maximally recoverable $(3g,r=3,h=3,a=1,q)$-local reconstruction code with a parity check matrix~(\ref{fig:MRtopology}), without loss of generality assume that for all $i\in [g],$
$$A_i=\begin{bmatrix}
1 & 1& 1\\
\end{bmatrix};\ B_i=\begin{bmatrix}
v^1_i&v^2_i&v^3_i\\
\end{bmatrix},$$
where $\{v^s_i\}_{s\in [3], i\in [g]}\subseteq \mathbb{F}_q^3.$ For each $i\in [g],$ define
$$
a_i = v^2_i - v^1_i \qquad b_i = v^3_i - v^2_i \qquad c_i = v^1_i - v^3_i.
$$
Clearly, for all $i\in [g],$ $a_i+b_i+c_i=0.$ Consider $\{a_i,b_i,c_i\}_{i\in [g]}$ as elements of $\PF_q^2$ and define our family to be $S=\cup_{i=1}^g \{a_i,b_i,c_i\}.$ It remains to show that all triples of elements of $S$ other than $\{a_i,b_i,c_i\}$ are non-collinear. When all three elements $v^\alpha_i - v^\beta_i, v^\gamma_j - v^\delta_j, v^\epsilon_k - v^\zeta_k$ belong to different groups this follows from the fact that, as implied by the MR property, the matrix
\begin{align*}
\left[
\begin{array}{cc|cc|cc}
1& 1 & 0 & 0& 0&0\\
\hline
0 & 0&1& 1 & 0&0\\
\hline
0 & 0&0& 0 & 1&1\\
\hline
v^\beta_i &v^\alpha_i & v^\delta_j & v^\gamma_j &  v^\zeta_k& v^\epsilon_k\\
\end{array}
\right]=
\left[
\begin{array}{cc|cc|cc}
1& 0 & 0 & 0& 0&0\\
\hline
0 & 0&1& 0 & 1&0\\
\hline
0 & 0&0& 0 & 0&0\\
\hline
v^\beta_i &v^\alpha_i - v^\beta_i & v^\delta_j & v^\gamma_j -v^\delta_j &  v^\zeta_k& v^\epsilon_k-v^\zeta_k\\
\end{array}
\right]
\end{align*}
is full rank. When triples come from two groups, (say, $v^\beta_i - v^\alpha_i, v^\gamma_i - v^\alpha_i, v^\delta_j - v^\epsilon_j$) this again follows from the MR property, as the matrix
\begin{align*}
\left[
\begin{array}{ccc|cc}
1& 1& 1 & 0 & 0\\
\hline
0 & 0&0& 1 & 1\\
\hline
v^\alpha_i& v^\beta_i &v^\gamma_i & v^\epsilon_j& v^\delta_j\\
\end{array}
\right] =
\left[
\begin{array}{ccc|cc}
1 & 0 & 0 & 0 & 0\\
\hline
0 & 0 & 0 & 1 & 0\\
\hline
v^\alpha_i& v^\beta_i - v^\alpha_i &v^\gamma_i - v^\alpha_i & v^\epsilon_j & v^\delta_j - v^\epsilon_j\\
\end{array}
\right]
\end{align*}
is also full rank.
\end{proof}
Combining Lemma~\ref{lem:matching_construction} and Lemma~\ref{lem:reduction_matching_triples} along with the fact that all the constructions are explicit gives the following theorem.
\begin{theorem}
For any $n>3$ which is a multiple of $3$ and for any finite field $\mathbb{F}_q,$ there exists an explicit maximally recoverable $(n,r=3,h=3,a=1,q)$-local reconstruction code provided that $q\geq \Omega \left(n\cdot \exp\left(C\sqrt{\log n}\right)\right)$ where $C>0$ is some absolute constant.
\end{theorem}

\subsection{Matching Collinear Triples from AP free sets}\label{Sec:matching_collinear}
In this section, we will prove Lemma~\ref{lem:matching_construction} by constructing a large $A\subset \PF_q^2$ with matching collinear triples. The main idea is to reduce the problem to constructing a large subset $A\subset \Z/N\Z$  with \emph{matching tri-sums} where $N=\Omega(q)$. A  subset $A\subset \Z/N\Z$ has matching tri-sums if $A$ can partitioned into disjoint triples, $A=\sqcup_i\{a_i,b_i,c_i\}$ such that the only 3 element subsets of $A$ which sum to zero are the triples $\{a_i,b_i,c_i\}$ in the partition. Such sets can be constructed from subsets of $[N]$ without any non-trivial arithmetic progressions.
The best known construction of a subset of $[N]$ with no non-trivial three term arithmetic progressions is due to Behrend~\cite{Beh46} which was slightly improved in~\cite{Elk11}. An explicit construction with similar bounds as~\cite{Beh46} was given in~\cite{Moser53}.

\begin{theorem}[\cite{Beh46,Moser53,Elk11}]\label{thm:Behrend}
For some absolute constant $C>0$, there exists an explicit $A\subset \{1,2,\cdots,N\}$ with $|A|\ge N\cdot \exp(-C\sqrt{\log N})$ which doesn't contain any 3 term arithmetic progressions i.e. there doesn't exist distinct $x,y,z\in A$ such that $x+z=2y$.
\end{theorem}
It is also known that any set $A\subset \{1,2,\cdots,N\}$ with no non-trivial 3 term arithmetic progressions should have size $|A|\lesssim \frac{(\log\log N)^4}{\log N}\cdot N$~\cite{Blo16}.

The reduction from matching collinear triples in $\F_q^2$ to subsets of $\Z/N\Z$ with matching tri-sums is simple when $q$ is a prime. In this case we can set $N=q$. Three points $(x_1,y_1),(x_2,y_2),(x_3,y_3)\in \F_q^2$ on the cubic curve $Y=X^3$ are collinear iff $x_1+x_2+x_3=0$. So we can get a large subset of $\PF_q^2$ with matching collinear triples, from a large subset of $\F_q\cong \Z/q\Z$ with matching tri-sums. And from Theorem~\ref{thm:Behrend}, we can get such a set of size $\ge q\cdot \exp(-O(\sqrt{\log q}))$.

When $q$ is not prime, the additive group of $\F_q$ is not cyclic anymore and subsets of $\F_q$ with matching tri-sums are much smaller. For example, if $\F_q$ has characteristic 2, which is the main setting of interest for us, the size of the largest subset of $\F_q$ with matching tri-sums is $\le q^c$ for some absolute constant $c<1$~\cite{Klein16}. We will use some results on elliptic curves which are a special kind of cubic curves to make the reduction work over any field.

\subsubsection{Elliptic curves}\label{SubSec:EliipticAndAP}
We will give a quick introduction to elliptic curves, please refer to~\cite{Sil09, Men13} for proofs and formal definitions. Let $\K$ be a finite field and $\bK$ be its algebraic closure. 
A singular Weierstrass equation\footnote{Usually elliptic curves are defined as curves given by non-singular Weierstrass equations. But for our purpose, it is easier to work with singular Weierstrass equations.} $E$ with singularity at $(X,Y,Z)=(0,0,1)$ is a cubic equation given by: $$E:\ Y^2Z+a_1XYZ-a_3X^2Z=X^3.$$ We associate with $E$ the set of all points in $\PbK^2$ which satisfy the equation $E$. There is exactly one point in $E$ with $Z$-coordinate equal to $0$, namely $(0:1:0)$, we call this special point \emph{the point at infinity} and denote it by $\cO$.
The set of non-singular $\K$-rational points of $E$, denoted by $E_{ns}(\K)$ is defined as follows:
\begin{align*}
E_{ns}(\K)=\{(x:y:1) | F(x,y,1)&=0,\ x,y\in \K,
(x,y)\ne (0,0) \}\cup \{\cO\}.
\end{align*}
$E_{ns}(\K)$ is an abelian group under a certain addition operation `$+$', with the point at infinity $\cO$ as the group identity. Under this operation, three points $a,b,c\in E_{ns}(\K)$ satisfy $a+b+c=\cO$ iff $a,b,c$ are collinear in $\PK^2$. The following theorem shows that $E_{ns}(\K)$ is isomorphic to $\K^*$ when $E$ is of a special form.

\begin{theorem}[Theorem 8.1 in~\cite{Men13}]
\label{Th:singular_weierstrass}
Let $E: (Y-\alpha X)(Y-\beta X)Z=X^3$ be a singular Weierstrass equation with $\alpha, \beta \in \K$ and $\alpha\ne \beta$. Then the map $\phi:E_{ns}(\K) \to \K^*$ defined as:
$$\phi:\cO\mapsto 1 \hspace{1cm} \phi:(x,y,1)\mapsto \frac{y-\beta x}{y-\alpha x}$$
is a group isomorphism.
\end{theorem}

Since $\K^*$ is a cyclic group for any finite field $\K$, $E_{ns}(\K)$ is isomorphic to $\Z/N\Z$ for $N=|\K|-1$ when $E$ is a singular Weierstrass equation as in Theorem~\ref{Th:singular_weierstrass}.

%

\subsubsection{Proof of Lemma~\ref{lem:matching_construction}}\label{SubSec:CollinearFamily}
\begin{proof}
Let $E$ be a singular Weierstrass equation\footnote{It is not essential to work with singular Weierstrass equations. The proof also works with non-singular elliptic curves as long as the group of $\K$-rational points is cyclic or has a large cyclic subgroup.} defined over $\F_q$ as in Theorem~\ref{Th:singular_weierstrass}. By Theorem~\ref{Th:singular_weierstrass}, $E_{ns}(\F_q)\cong \Z/N\Z$ where $N=q-1$. Recall that $a,b,c\in E_{ns}(\F_q)$ satisfy $a+b+c=\cO$ in the group iff they are collinear. 

Let $B\subset \{1,2,\cdots,N/20\}$ be an explicit subset of size $|B|\gtrsim N\cdot \exp(-C\sqrt{\log N})$ with no 3-term arithmetic progressions, as guaranteed by Theorem~\ref{thm:Behrend}. Now define subsets $A_1,A_2,A_3\subset \Z/N\Z$ as 
\begin{align*}
A_1=\left\{x: x\in B\right\}, A_2=\left\{\left\lfloor\frac{N}{3}\right\rfloor+x : x\in B\right\},
A_3=\left\{N-\left\lfloor\frac{N}{3}\right\rfloor-2x : x\in B\right\}.
\end{align*}
Clearly, $A_1,A_2,A_3$ are disjoint. Finally we define $\tilde{A}=A_1\cup A_2\cup A_3$.
Now we claim that the only triples from $\tilde{A}$ which sum to zero in $\Z/N\Z$ are $\{x,\floor{N/3}+x,N-\floor{N/3}-2x\}$ for $x\in B$ and these triples form a partition of $\tilde{A}$.

It is not hard to see that if three distinct elements $a,b,c\in \tilde{A}$ satisfy $a+b+c=0,$ then $a,b,c$ should come from 3 different sets $A_1,A_2,A_3$. So after reordering, we can assume $$a=x,b=\floor{N/3}+y,c=N-\floor{N/3}-2z$$ for some $x,y,z\in B$.
Thus $a+b+c=0$ implies that $x+y=2z$, which implies that $x=y=z$ since $B$ is free from 3 arithmetic progressions.

Finally let $A\subset \PF_q^2$ be the set of points in $E_{ns}(\F_q)$ which map to the set $\tilde{A}\subset \Z/N\Z$ under the isomorphism $E_{ns}(\F_q)\cong \Z/N\Z$. Now it is easy to see that $A$ has matching collinear triples and we have $|A|\gtrsim q\cdot \exp(-C\sqrt{\log q})$.
\end{proof}

\section{Open problems}\label{Sec:Open}
In this work we made progress towards quantifying the minimal size of finite fields required for existence of maximally recoverable local reconstruction codes and obtained both lower and upper bounds. There is a wide array of questions that remain open. Here we highlight some of them:
\begin{itemize}
\itemsep=0ex
\item Our lower bound~(\ref{Eqn:MainLowerFormOur_mega}) implies that even in the regime of constant $a$ and $h,$ when $h\geq 3, a\ge 1$ and $r$ grows with $n$ there exist no MR codes over fields of size $O(n).$ It would be of great interest to understand if such codes always exist when all parameters $a,h,$ and $r$ are held constant and only $n$ grows.

\item Our lower bound~(\ref{Eqn:MainLowerFormOur_mega}) is of the form $q= \Omega(nr^\alpha)$ where $\alpha>0$ in all parameter ranges except when $a=0$ or $h=2$ or $g=2$ or $(g=3, h=4,a=1)$. When $a=0$ or $h=2$, we now know that there are linear field size constructions for any $r$. Is this also true when $g=2$?

\item In the case of fields of characteristic two, can one reduce the field sizes in Theorems~\ref{Th:H2Main} and~\ref{Th:H3Main} to $O(n)$ and $O(n^3)$ to match the case of prime fields?

\item Our Lemma~\ref{lem:reduction_matching_triples} provides an equivalence between the parameters of families of matching collinear triples in the projective plane and maximally recoverable local reconstruction codes with $r=3,h=3,$ and $a=1.$ We hope that this reduction will be useful to obtain an $\omega(n)$ lower bound for the alphabet size of MR $(n,r=3,h=3,a=1,q)$-LRCs, or lead to a construction over fields of linear size. It is also very interesting to see if techniques similar to those in Section~\ref{Sec:matching_collinear} can be used to get codes over fields of nearly linear size when $r>3$ or $a>1$ or $h>3.$

\item Finally, it is interesting to see if our lower bound in Theorem~\ref{Th:lowerbound_mega} can be generalized to the setting of non-linear codes. Basic results about LRCs such as distance vs. redundancy trade-off~\cite{GHSY} have been generalized to non-linear setting in~\cite{XOR_ELE,FY}.
\end{itemize}

\section*{Acknowledgements}
We thank Madhu Sudan his very useful suggestion to use pairwise independence properties of hyperplanes to prove Lemma~\ref{lem:hyperplane_incidence}. We would like to thank Parikshit Gopalan for allowing us to include his Proposition~\ref{prop:MainParikshit} in this paper. 

We are grateful to Cheng Huang for asking the question that led us to start this project, Suryateja Gavva and Ilya Shkredov for helpful discussions about this work.   


\bibliographystyle{alpha}
\bibliography{references}

\newcommand{\etalchar}[1]{$^{#1}$}
\begin{thebibliography}{DGW{\etalchar{+}}10}

\bibitem[Bal12]{MainMDS1}
Simeon Ball.
\newblock On sets of vectors of a finite vector space in which every subset of
  basis size is a basis.
\newblock {\em Journal of European Mathematical Society}, 14:733--748, 2012.

\bibitem[Beh46]{Beh46}
Felix~A Behrend.
\newblock On sets of integers which contain no three terms in arithmetical
  progression.
\newblock {\em Proceedings of the National Academy of Sciences},
  32(12):331--332, 1946.

\bibitem[BFI86]{BFI86}
Enrico Bombieri, John~B Friedlander, and Henryk Iwaniec.
\newblock Primes in arithmetic progressions to large moduli.
\newblock {\em Acta Mathematica}, 156(1):203--251, 1986.

\bibitem[BHH13]{BHH}
Mario Blaum, James~Lee Hafner, and Steven Hetzler.
\newblock Partial-{MDS} codes and their application to {RAID} type of
  architectures.
\newblock {\em IEEE Transactions on Information Theory}, 59(7):4510--4519,
  2013.

\bibitem[Bla13]{Blaum}
Mario Blaum.
\newblock Construction of {PMDS} and {SD} codes extending {RAID} 5.
\newblock Arxiv 1305.0032, 2013.

\bibitem[Blo16]{Blo16}
Thomas~F Bloom.
\newblock A quantitative improvement for {R}oth's theorem on arithmetic
  progressions.
\newblock {\em Journal of the London Mathematical Society}, page jdw010, 2016.

\bibitem[BPSY16]{BPSY}
Mario Blaum, James Plank, Moshe Schwartz, and Eitan Yaakobi.
\newblock Construction of partial {MDS} and sector-disk codes with two global
  parity symbols.
\newblock {\em IEEE Transactions on Information Theory}, 62(5):2673--2681,
  2016.

\bibitem[CHL07]{CHL}
Minghua Chen, Cheng Huang, and Jin Li.
\newblock On maximally recoverable property for multi-protection group codes.
\newblock In {\em IEEE International Symposium on Information Theory (ISIT)},
  pages 486--490, 2007.

\bibitem[CK17]{CK}
Gokhan Calis and Ozan Koyluoglu.
\newblock A general construction fo {PMDS} codes.
\newblock {\em IEEE Communications Letters}, 21(3):452--455, 2017.

\bibitem[DGW{\etalchar{+}}10]{Dimakis_1}
Alexandros~G. Dimakis, Brighten Godfrey, Yunnan Wu, Martin~J. Wainwright, and
  Kannan Ramchandran.
\newblock Network coding for distributed storage systems.
\newblock {\em IEEE Transactions on Information Theory}, 56(9):4539--4551,
  2010.

\bibitem[Elk11]{Elk11}
Michael Elkin.
\newblock An improved construction of progression-free sets.
\newblock {\em Israel journal of mathematics}, 184(1):93--128, 2011.

\bibitem[FY14]{FY}
Michael Forbes and Sergey Yekhanin.
\newblock On the locality of codeword symbols in non-linear codes.
\newblock {\em Discrete mathematics}, 324:78--84, 2014.

\bibitem[GHJY14]{GHJY}
Parikshit Gopalan, Cheng Huang, Bob Jenkins, and Sergey Yekhanin.
\newblock Explicit maximally recoverable codes with locality.
\newblock {\em IEEE Transactions on Information Theory}, 60(9):5245--5256,
  2014.

\bibitem[GHK{\etalchar{+}}17]{GHKSWY}
Parikshit Gopalan, Guangda Hu, Swastik Kopparty, Shubhangi Saraf, Carol Wang,
  and Sergey Yekhanin.
\newblock Maximally recoverable codes for grid-like topologies.
\newblock In {\em 28th Annual Symposium on Discrete Algorithms (SODA)}, pages
  2092--2108, 2017.

\bibitem[GHSY12]{GHSY}
Parikshit Gopalan, Cheng Huang, Huseyin Simitci, and Sergey Yekhanin.
\newblock On the locality of codeword symbols.
\newblock {\em IEEE Transactions on Information Theory}, 58(11):6925 --6934,
  2012.

\bibitem[GJX18]{GLX-ff}
Venkatesan Guruswami, Lingfei Jin, and Chaoping Xing.
\newblock Constructions of maximally recoverable local reconstructon codes via
  function fields.
\newblock Manuscript, 2018.

\bibitem[Gop17]{Gopalan}
Parikshit Gopalan.
\newblock Personal communication, 2017.

\bibitem[GW16]{GW}
Venkatesan Guruswami and Mary Wootters.
\newblock Repairing {R}eed-{S}olomon codes.
\newblock In {\em 48th ACM Symposium on Theory of Computing (STOC)}, pages
  216--226, 2016.

\bibitem[GYBS17]{GYBS}
Ryan Gabrys, Eitan Yaakobi, Mario Blaum, and Paul Siegel.
\newblock Construction of partial {MDS} codes over small finite fields.
\newblock In {\em 2017 IEEE International Symposium on Information Theory
  (ISIT)}, pages 1--5, 2017.

\bibitem[HCL07]{HCL}
Cheng Huang, Minghua Chen, and Jin Li.
\newblock Pyramid codes: flexible schemes to trade space for access efficiency
  in reliable data storage systems.
\newblock In {\em 6th IEEE International Symposium on Network Computing and
  Applications (NCA 2007)}, pages 79--86, 2007.

\bibitem[HSX{\etalchar{+}}12]{HuangSX}
Cheng Huang, Huseyin Simitci, Yikang Xu, Aaron Ogus, Brad Calder, Parikshit
  Gopalan, Jin Li, and Sergey Yekhanin.
\newblock Erasure coding in {W}indows {A}zure {S}torage.
\newblock In {\em USENIX Annual Technical Conference (ATC)}, pages 15--26,
  2012.

\bibitem[HY16]{HY}
Guangda Hu and Sergey Yekhanin.
\newblock New constructions of {SD} and {MR} codes over small finite fields.
\newblock In {\em 2016 IEEE International Symposium on Information Theory
  (ISIT)}, pages 1591--1595, 2016.

\bibitem[Kle16]{Klein16}
Robert Kleinberg.
\newblock A nearly tight upper bound on tri-colored sum-free sets in
  characteristic 2.
\newblock {\em arXiv preprint arXiv:1605.08416}, 2016.

\bibitem[KLR17]{KLR}
Daniel Kane, Shachar Lovett, and Sankeerth Rao.
\newblock Labeling the complete bipartite graph with no zero cycles.
\newblock In {\em 58th IEEE Symposium on Foundations of Computer Science
  (FOCS)}, 2017.

\bibitem[LN83]{LN}
Rudolf Lidl and Harald Niederreiter.
\newblock {\em Finite Fields}.
\newblock Cambridge University Press, Cambridge, 1983.

\bibitem[Lov18]{Lovett18}
Shachar Lovett.
\newblock A proof of the {GM-MDS} conjecture.
\newblock {\em Electronic Colloquium on Computational Complexity {(ECCC)}},
  25:47, 2018.

\bibitem[MBG{\etalchar{+}}13]{Men13}
A.J. Menezes, I.F. Blake, X.H. Gao, R.C. Mullin, S.A. Vanstone, and
  T.~Yaghoobian.
\newblock {\em Applications of Finite Fields}.
\newblock The Springer International Series in Engineering and Computer
  Science. Springer US, 2013.

\bibitem[Mos53]{Moser53}
Leo Moser.
\newblock {\em On non-averaging sets of integers}.
\newblock Canadian Mathematical Society, 1953.

\bibitem[MS77]{MS}
F.~J. MacWilliams and N.~J.~A. Sloane.
\newblock {\em The Theory of Error Correcting Codes}.
\newblock North Holland, Amsterdam, New York, 1977.

\bibitem[PD14]{Dimakis_0}
Dimitris Papailiopoulos and Alexandros Dimakis.
\newblock Locally repairable codes.
\newblock {\em IEEE Transactions on Information Theory}, 60(10):5843--5855,
  2014.

\bibitem[PGM13]{Plank}
J.~S. Plank, K.~M. Greenan, and E.~L. Miller.
\newblock Screaming fast {G}alois field arithmetic using {I}ntel {SIMD}
  instructions.
\newblock In {\em 11th Usenix Conference on File and Storage Technologies
  (FAST)}, pages 299--306, San Jose, February 2013.

\bibitem[SAP{\etalchar{+}}13]{XOR_ELE}
Maheswaran Sathiamoorthy, Megasthenis Asteris, Dimitris~S. Papailiopoulos,
  Alexandros~G. Dimakis, Ramkumar Vadali, Scott Chen, and Dhruba Borthakur.
\newblock {XOR}ing elephants: novel erasure codes for big data.
\newblock In {\em Proceedings of VLDB Endowment (PVLDB)}, pages 325--336, 2013.

\bibitem[Sil09]{Sil09}
J.H. Silverman.
\newblock {\em The Arithmetic of Elliptic Curves}.
\newblock Graduate Texts in Mathematics. Springer New York, 2009.

\bibitem[TB14]{TB}
Itzhak Tamo and Alexander Barg.
\newblock A family of optimal locally recoverable codes.
\newblock {\em IEEE Transactions on Information Theory}, 60:4661--4676, 2014.

\bibitem[TPD16]{TPD}
Itzhak Tamo, Dimitris Papailiopoulos, and Alexandros~G. Dimakis.
\newblock Optimal locally repairable codes and connections to matroid theory.
\newblock {\em IEEE Transactions on Information Theory}, 62:6661--6671, 2016.

\bibitem[WTB17]{WTB}
Zhiying Wang, Itzhak Tamo, and Jehoshua Bruck.
\newblock Optimal rebuilding of multiple erasures in {MDS} codes.
\newblock {\em IEEE Transactions on Information Theory}, 63:1084--1101, 2017.

\bibitem[YB17]{YB}
Min Ye and Alexander Barg.
\newblock Explicit constructions of high-rate {MDS} array codes with optimal
  repair bandwidth.
\newblock {\em IEEE Transactions on Information Theory}, 63:2001--2014, 2017.

\bibitem[Yek12]{Y_now}
Sergey Yekhanin.
\newblock Locally decodable codes.
\newblock {\em Foundations and trends in theoretical computer science},
  6(3):139--255, 2012.

\bibitem[YH18]{YH18}
Hikmet Yildiz and Babak Hassibi.
\newblock Optimum linear codes with support constraints over small fields.
\newblock {\em CoRR}, abs/1803.03752, 2018.

\end{thebibliography}

\appendix

\section{Proof of Proposition~\ref{prop:LowerAsymptotic_gsmall}}
\label{Sec:proofof_lowerasymptotic_gsmall}
We will first focus on the case when $a\le h-2\ceil{h/g}$ and later in Proposition~\ref{prop:Mainlowerbound_gsmall_alarge} we will deal with the case $a>h-2\ceil{h/g}$.
\begin{proposition}
\label{prop:Mainlowerbound_gsmall_asmall}
Suppose $a,g,h$ be fixed constants such that $2\le g\le h$ and $a \le h-2\ceil{h/g}$. Let $C$ be a maximally recoverable $(n,r,h,a,q)$-LRC where $r=n/g$ is the size of each local group. Then $$q\ge \Omega_{a,h,g}(n^{1+a/\ceil{h/g}}).$$
\end{proposition}
\begin{proof}
Let $t_1\ge t_2 \ge \cdots \ge t_g$ be such that $t_i=\ceil{h/g}$ or $t_i=\floor{h/g}$ and $\sum_{i=1}^g t_i=h$. Given a matrix $M$, we will denote its kernel by $\ker(M)=\{x: Mx=0\}$ and its image by $\img(M)=\{y: \exists x\text{ s.t. } Mx=y\}$. We call the subspace spanned by the rows of $M$ as the row space of $M$ and the subspace spanned by the columns of $M$ as the column space of $M$ and their dimensions are both equal to $\rk(M)$. Note that $\img(M)$ is equal to the column space of $M$ and $\ker(M)$ is the orthogonal subspace of the row space of $M$. $M^\perp$ is defined as a matrix with independent columns such that $\img(M^\perp)=\ker(M)$ and so $MM^\perp=0$. Note that $M^\perp$ is not unique, any matrix whose columns span $\ker(M)$ can be used as $M^\perp$, but the specific choice of $M^\perp$ is not important for the proof.

  According to the discussion in Section~\ref{Sec:Prelim} the code $C$ admits a parity check matrix of the shape
\begin{equation}
\label{eqn:MRparitycheck_2}
\left[
\begin{array}{c|c|c|c}
A_1 & 0 & \cdots & 0\\
\hline
0 &A_2 & \cdots & 0\\
\hline
\vdots & \vdots & \ddots & \vdots \\
\hline
0 & 0 & \cdots & A_g \\
\hline
B_1 & B_2 & \cdots & B_g \\
\end{array}
\right].
\end{equation}
Here $A_1,A_2,\cdots,A_g$ are $a\times r$ matrices over $\F_q$, $B_1,B_2,\cdots,B_g$ are $h\times r$ matrices over $\F_q.$ The rest of the matrix is filled with zeros. Every $a\times a$ minor in each matrix $\{A_i\}_{i\in [g]}$ has full rank. So for every subset $S\subseteq [r]$ of size $|S|=a+t_i,$ the matrix $A_i(S)$ is an $a\times(a+t_i)$ matrix of full rank. Let $A_i(S)^\perp$ be an $(a+t_i)\times t_i$ matrix of full rank such that $A_i(S)A_i(S)^\perp =0$ (note that $A_i(S)^\perp$ is not unique). Now define $$P_{i,S}=B_i(S)A_i(S)^\perp$$ which is a $h\times t_i$ matrix.

Define $p_{i,S}$ as the subspace of $\F_q^h$ spanned by the columns of $P_{i,S}$. The MR property implies that any subset of columns of the parity check matrix~(\ref{eqn:MRparitycheck_2}) which can be obtained by picking $a$ columns in each local group and $h$ arbitrary additional columns is full rank. We will use this property to make two claims about the subspaces $\left\{p_{i,S}\right\}.$

\begin{claim}\label{claim:diff_groups_2}
For every subsets $S_1,\cdots,S_g\subseteq [r]$ such that $|S_i|=a+t_i$, the spaces $p_{1,S_1},\dots,p_{g,S_g}$ together span the entire space i.e. $p_{1,S_1}\oplus p_{2,S_2}\oplus \cdots \oplus p_{g,S_g}=\F_q^h$.
\end{claim}
\begin{proof}
Consider the following matrix equation:
\begin{equation}
\left[
\begin{array}{c|c|c|c}
A_{\ell_1}(S_1) & 0 & \cdots & 0\\
\hline
0 &A_{\ell_2}(S_2) & \cdots & 0\\
\hline
\vdots & \vdots & \ddots & \vdots \\
\hline
0 & 0 & \cdots & A_{\ell_h}(S_h) \\
\hline
B_{\ell_1}(S_1) & B_{\ell_2}(S_2) & \cdots & B_{\ell_h}(S_h) \\
\end{array}
\right]
\left[
\begin{array}{c|c|c|c}
A_{\ell_1}(S_1)^\perp & 0 & \cdots & 0\\
\hline
0 &A_{\ell_2}(S_2)^\perp & \cdots & 0\\
\hline
\vdots & \vdots & \ddots & \vdots \\
\hline
0 & 0 & \cdots & A_{\ell_h}(S_h)^\perp \\
\end{array}
\right]
=
\left[
\begin{array}{c|c|c|c}
0 & 0 & \cdots & 0\\
\hline
0 &0 & \cdots & 0\\
\hline
\vdots & \vdots & \ddots & \vdots \\
\hline
0 & 0 & \cdots & 0 \\
\hline
P_{\ell_1,S_1} & P_{\ell_2,S_2} & \cdots & P_{\ell_h,S_h}\\
\end{array}
\right].
\end{equation}

Let us denote the matrices in the above equation by $M_1,M_2,M_3$ such that the above equation becomes $M_1M_2=M_3$. By MR property, when we erase the coordinates corresponding to $S_1,\cdots,S_g$ in groups $1,\cdots,g$ respectively, the resulting erasure pattern is correctable. This implies that the $(ag+h)\times (ag+h)$ matrix $M_1$ is full rank. Also $M_2$ has full column rank because of its block structure. So $M_3$, which is an $(ag+h)\times h$ matrix, should have full column rank. This proves the required statement since $p_{i,S}$ is the column space of $P_{i,S}$.
\end{proof}

The above claim in particular implies that the matrices $P_{i,S}$ have full rank and that $p_{i,S}$ is a $t_i$-dimensional subspace of $\F_q^h$ for every $i$ and $S$. The following claim explains for a fixed $i$, how subspaces $\{p_{i,S}:|S|=a+t_i\}$ intersect with each other.
 
 \begin{claim}
 \label{claim:same_groups_2}
 Let $i\in [g]$ and $S,T$ be subsets of $[r]$ of size $a+t_i$ such that $|S\cap T|=\ell$.
 \begin{enumerate}
 \item If $\ell \le a$ then $p_{i,S}\cap p_{i,T}=\phi$.
 \item If $\ell =a+\ell'$ for $\ell' \ge 1$ then $\dim(p_{i,S}\cap p_{i,T})=\ell'$.
 \end{enumerate}
\end{claim}
\begin{proof}
Consider the following matrix equation:
\begin{equation}
\left[
\begin{array}{c|c}
A_i(S) & A_i(T)\\
\hline
B_i(S) & B_i(T)\\
\end{array}
\right]
\left[
\begin{array}{c|c}
A_i(S)^\perp & 0\\
\hline
0 & A_i(T)^\perp\\
\end{array}
\right]
=
\left[
\begin{array}{c|c}
0 & 0\\
\hline
P_{i,S} & P_{i,T}\\
\end{array}
\right].
\end{equation}
Let us denote the matrices that appear in the above equation to be $M_1,M_2,M_3$ in that order so that above equation becomes $M_1M_2=M_3$. The matrix $M_1$ is an $(a+h)\times 2(a+t_i)$ matrix of rank $|S\cup T|=2(a+t_i)-\ell$. This is because any $a+h$ columns of $\mattwoone{A_i}{B_i}$ are linearly independent by MR property and $|S\cup T|\le 2(a+t_i) \le a+h$  by the assumption that $a\le h-2\ceil{h/g}$. Wlog, we can reorder the columns of $M_1$ such that the first $\ell$ columns of $\mattwoone{A_i(S)}{B_i(S)}$ and $\mattwoone{A_i(T)}{B_i(T)}$ are identical.
$M_2$ is an $2(a+t_i)\times 2t_i$ matrix of full rank. $M_3$ is an $(a+h)\times 2t_i$ matrix and $\dim(p_{i,S}\cap p_{i,T}) = 2t_i-\rk(M_3)=\dim(\ker(M_3)).$ Since $\ker(M_2)=\phi$, $$\dim(p_{i,S}\cap p_{i,T})=\dim(\ker(M_3))=\dim(\img(M_2)\cap \ker(M_1)).$$

\noindent
\textbf{Case 1:} $|S\cap T|=\ell \le a$

\noindent
We need to show that $\img(M_2)\cap \ker(M_1)=\phi$. Suppose there is a non-zero vector in $\img(M_2)\cap \ker(M_1)$, say $\beta$. We completely understand the kernel of $M_1$, the only linear dependencies of the columns of $M_1$ occur because of repetitions i.e. $$\mathrm{ker}(M_1)=\linearspan{e_1-e_{a+t_i+1},\dots,e_\ell-e_{a+t_i+\ell}}.$$ 
So the first half of $\beta$ is a non-zero vector in $\img(A_i(S)^\perp)=\ker(A_i(S))$ which is supported on the first $\ell$ coordinates. But we know that any $a$ columns of $A_i(S)$ are linearly independent and so its kernel cannot contain any non-zero $\ell$-sparse vector when $\ell \le a$, leading to a contradiction.

\smallskip \noindent
\textbf{Case 2:} $|S\cap T| =\ell= a+\ell'$

\noindent We need to show that $\dim(\img(M_2)\cap \ker(M_1))=\ell'$.
\begin{itemize}
\item We will first show  that $\dim(\img(M_2)\cap \ker(M_1))\ge \ell'$.

We will exhibit $\ell'$ linearly independent vectors in $\img(M_2)\cap \ker(M_1)$. The first $a$ columns of $A_i(S)$ are linearly independent. So the next $\ell'$ columns of $A_i(S)$ can be written as linear combinations of them. This gives $\ell'$ linearly independent vectors in $\ker(A_i(S))=\img(A_i(S)^\perp)$, call them $\alpha_1,\dots,\alpha_{\ell'}$.  Since the first $a+\ell'$ columns of $A_i(S)$ and $A_i(T)$ are the same, the vectors $\alpha_1,\dots,\alpha_{\ell'}$ are also in $\ker(A_i(T))=\img(A_i(T)^\perp)$. Thus the vectors $\mattwoone{\alpha_1}{-\alpha_1},\cdots, \mattwoone{\alpha_{\ell'}}{-\alpha_{\ell'}}$ are in the column space of $M_2$. But since $\alpha_1,\cdots, \alpha_{\ell'}$ are supported on the first $a+\ell'$ coordinates and the first $a+\ell'$ columns of $\mattwoone{A_i(S)}{B_i(S)}$ and $\mattwoone{A_i(T)}{B_i(T)}$ are identical, it is easy to see that $\mattwoone{\alpha_1}{-\alpha_1},\cdots, \mattwoone{\alpha_{\ell'}}{-\alpha_{\ell'}}$ are in the kernel of $M_1$. Moreover these vectors are linearly independent because $\alpha_1,\cdots,\alpha_{\ell'}$ are linearly independent. This proves that $\dim(\img(M_2)\cap \ker(M_1)) \ge \ell'$.

\item We now show that $\dim(\img(M_2)\cap \ker(M_1))\le \ell'$.\\
Suppose $\dim(\img(M_2)\cap \ker(M_1))=\ell'' \ge \ell'+1$. So $\img(M_2)\cap \ker(M_1)$ contains a non-zero vector, say $\beta$, whose first $\ell''-1$ coordinates are zero. Since 
$$\beta \in \mathrm{ker}(M_1)=\linearspan{e_1-e_{a+t_i+1},\dots,e_\ell-e_{a+t_i+\ell}},$$ and the first $\ell''-1$ coordinates of $\beta$ are zero, $$\beta\in \linearspan{e_{\ell''}-e_{a+t_i+\ell''},\dots,e_{\ell}-e_{a+t_i+\ell}}.$$ Since $\beta\in \img(M_2)$, the first half of $\beta$ is a non-zero vector in $\img(A_i(S)^\perp)$  supported on $\ell-(\ell''-1)\le a$ coordinates. This is a contradiction because any $a$ columns of $A_i(S)$ are linearly independent and thus $\img(A_i(S)^\perp)=\ker(A_i(S))$ cannot contain a non-zero $a$-sparse vector. \qedhere
\end{itemize}
\end{proof}

Now we will show that if $q=o_{a,g,h}(n^{1+a/\ceil{h/g}})$ then a random $(h-1)$-dimensional subspace of $\F_q^h$ will contain $p_{1,S_1},p_{2,S_2},\dots,p_{g,S_g}$ for some subsets $S_1,\dots,S_g \subset [r]$ with $|S_i|=a+t_i$ with high probability, which contradicts Claim~\ref{claim:diff_groups_2}. Let $f$ be a uniformly random vector in $\F_q^h$ and let $F=\{x\in F_q^h: \inpro{x}{f}=0\}$ i.e. the set of vectors orthogonal to $f$. If $f\ne 0$, then $F$ is a $(h-1)$-dimensional subspace and if $f=0$ then $F=\F_q^h$. We want to calculate the probability that $F$ contains $p_{1,S_1},p_{2,S_2},\dots,p_{g,S_g}$ for some subsets $S_1,\dots,S_g$ conditioned on $F$ not being the entire space i.e. $f\ne 0$. Let's ignore the conditioning for now and estimate the required probability.

Fix some $i\in [g]$. Let $Z_i$ be the number of subspaces among $\{p_{i,S}: S\in \binom{[r]}{a+t_i}\}$ which are contained in $F$. We have $\Pr[Z_i>0]\ge \E[Z_i]^2/\E[Z_i^2].$ The probability that $F$ contains a fixed $p_{i,S}$ which is a $t_i$-dimensional subspace is $1/q^{t_i}$. Therefore,
$$\E[Z_i]=\sum_{S\subset [r], |S|=a+t_i} \Pr[p_{i,S}\in F]=\frac{\binom{r}{a+t_i}}{q^{t_i}}.$$ 
\begin{align*}
\E[Z_i^2]&=\sum_{S,T\in \binom{r}{a+t_i}} \Pr[p_{i,S},p_{i,T}\in F]\\
&=\sum_{\ell=0}^a\sum_{S,T: |S\cap T|=\ell} \Pr[p_{i,S},p_{i,T}\in F] + \sum_{\ell'=1}^{t_i} \sum_{S,T: |S\cap T|=a+\ell'} \Pr[p_{i,S},p_{i,T}\in F].
\end{align*}
By Claim~\ref{claim:same_groups_2}, if $|S\cap T|\le a$, then $p_{i,S}\cap p_{i,T}=\phi$ and so $$\Pr[p_{i,S},p_{i,T}\in F]=\frac{1}{q^{2t_i}}.$$ And if $|S\cap T|=a+\ell'$ then $\dim(p_{i,S}\cap p_{i,T})=\ell'$ and so $$\Pr[p_{i,S},p_{i,T}\in F]=\frac{1}{q^{2t_i-\ell'}}.$$ Therefore,
\begin{align*}
\E[Z_i^2]=\sum_{\ell=0}^a \binom{r}{a+t_i}\binom{r-(a+t_i)}{a+t_i-\ell}\binom{a+t_i}{\ell} \frac{1}{q^{2t_i}} + \sum_{\ell'=0}^{t_i} \binom{r}{a+t_i}\binom{r-(a+t_i)}{t_i-\ell'}\binom{a+t_i}{a+\ell'} \frac{1}{q^{2t_i-\ell'}}.
\end{align*}
Therefore,
\begin{align*}
\frac{\E[Z_i^2]}{\E[Z_i]^2}=1 + \sum_{\ell'=1}^{t_i} (c_{\ell'}+o_{a,g,h}(1)) \frac{q^{\ell'}}{n^{a+\ell'}}+o_{a,g,h}(1)
\end{align*}
where $c_{\ell'}$ are constants depending only on $a,g,h$ and indepedent of $n,q$.

When $q=o_{a,g,h}(n^{1+a/t_i})$, which is true since $t_i\le \ceil{h/g}$, $E[Z_i^2]/E[Z_i]^2 =1+o(1)$ and so $\Pr[Z_i>0]=1-o(1)$. By union bound, $\Pr[\forall i\in [g], Z_i>0]=1-o(1)$. Note that $q$ should grow with $n$ to have enough subspaces for  Claim~\ref{claim:same_groups_2} to hold. Therefore $\Pr[f=0]=1/q^h=o(1).$ So $$\Pr\left[\forall i\in [g], Z_i>0 \big| f\ne 0\right]\ge \Pr[\forall i\in [g], Z_i>0]-\Pr[f=0]=1-o(1)$$ which implies the required contradiction.
\end{proof}

Using the suggestion of Parikshit Gopalan~\cite{Gopalan}, we can generalize Proposition~\ref{prop:Mainlowerbound_gsmall_asmall} to the case when $a>h-2\ceil{h/g}.$ In this case, we modify the proof of Proposition~\ref{prop:Mainlowerbound_gsmall_asmall} where we only consider sets $S_i$ that have size $a+t_i$ but are constrained to contain the set $\{1,2,\ldots,a+2t_i-h\},$ as this ensures that pairwise unions still have size at most $a+h.$ Clearly, the total number of such sets is $\binom{r-a+h-2t_i}{h-t_i}.$ The rest of the proof remains the same and yields the following:
\begin{proposition}\label{prop:Mainlowerbound_gsmall_alarge}
Assume $a,h,g$ are fixed constants such that $a\ge h-2\ceil{h/g}$ and $h\ge g\ge 2$, then any maximally recoverable $(n,r,h,a,q)$-local reconstruction code with $g=n/r$ local groups must have
\begin{equation}\label{Eqn:Mainlowerbound_gsmall_alarge}
q \geq \Omega_{a,h,g}(n^{h/\ceil{h/g}-1}).
\end{equation}
\end{proposition}

\begin{proof}[Proof of Proposition~\ref{prop:LowerAsymptotic_gsmall}]
Follows immediately from Propositions~\ref{prop:Mainlowerbound_gsmall_asmall} and ~\ref{prop:Mainlowerbound_gsmall_alarge}.
\end{proof}

\section{Determinantal identities}
\label{Sec:determinantal}
For our constructions, we will need some determinantal identities which we prove here. We need the following expansion of determinant of a column partitioned matrix.
\begin{lemma}
\label{lem:Laplace_expansion}
For $i\in [\ell]$, let $F_i$ be an $h\times t_i$ matrix with $\sum_{i=1}^\ell t_i=h$. Then,
$$\det[F_1|F_2|\cdots|F_\ell]=\sum_{S_1\sqcup \cdots \sqcup S_\ell =[h], |S_i|=t_i} \sgn(S_1,\cdots,S_\ell) \prod_{i\in [\ell]} \det F_i^{(S_i)}$$ where $S_1\sqcup \dots \sqcup S_\ell$ ranges over partitions of $[h]$ such that $|S_i|=t_i$. Here $\sgn(S_1,\cdots,S_\ell)$ is the sign of the permutation taking $(1,2,\cdots,h)$ to $(\tilde{S}_1,\tilde{S}_2,\cdots,\tilde{S}_\ell)$ where $\tilde{S_i}$ is the tuple formed by ordering the elements of $S_i$ in increasing order.
\end{lemma}
\begin{proof}
Given distinct integers $a_1,\cdots,a_n$, define $\sgn(a_1, a_2, \cdots, a_n):=(-1)^t$ where $t$ is number of transpositions needed to sort the elements $a_1,a_2,\cdots,a_n$ in increasing order. Thus for a permutation $\pi\in S_h$, $\sgn(\pi)=\sgn(\pi(1),\pi(2),\cdots,\pi(h))$. Let $F=[F_1|F_2|\cdots|F_\ell]$ and  for $i\in [\ell]$, let $T_i=\{t_{i-1}+1,\cdots,t_i\}$ where $t_0=0$. We can expand $\det(F)$ as:
\begin{align*}
\det(F)&=\sum_{\pi\in S_h} \sgn(\pi) \prod_{i=1}^h F_{\pi(i)i}\\
&=\sum_{S_1\sqcup \cdots \sqcup S_\ell =[h], |S_i|=t_i} \ \ \sum_{\pi:\ \pi(T_i)=S_i} \sgn(\pi) \prod_{i=1}^h F_{\pi(i)i}\\
\end{align*}
Note that if $\pi(T_i)=S_i,$ then for $i\in [\ell]$,
$$\sgn(\pi)=\sgn(\tilde{S_1},\cdots,\tilde{S_\ell}) \prod_{i=1}^\ell \sgn(\pi(t_{i-1}+1),\cdots,\pi(t_i))$$
because we can sort $(\pi(1),\cdots,\pi(h))$ first within each group to get $(\tilde{S_1},\cdots,\tilde{S_\ell})$ and then sort it to get $(1,2,\cdots,h).$ Therefore,
\begin{align*}
&\sum_{\pi:\ \pi(T_i)=S_i} \sgn(\pi) \prod_{i=1}^h F_{\pi(i)i} \\
&= \sum_{\sigma_1:T_1\to S_1, \dots,\ \sigma_\ell:T_\ell \to S_\ell} \sgn(\tilde{S_1},\cdots,\tilde{S_\ell}) \prod_{i=1}^\ell \left(\sgn(\sigma_i(t_{i-1}+1),\cdots,\sigma_i(t_i)) \prod_{j=t_{i-1}+1}^{t_i} F_{\sigma_i(j)j} \right)\\
&\hspace{3cm}\tag{where the summation is over all bijections $\sigma_i:T_i\to S_i$}\\
&= \sgn(\tilde{S_1},\cdots,\tilde{S_\ell}) \prod_{i=1}^\ell \left(\sum_{\sigma_i:T_i\to S_i}  \sgn(\sigma_i(t_{i-1}+1),\cdots,\sigma_i(t_i)) \prod_{j=t_{i-1}+1}^{t_i} F_{\sigma_i(j)j}\right)\\
&= \sgn(\tilde{S_1},\cdots,\tilde{S_\ell}) \prod_{i=1}^\ell \det F_i^{(S_i)}. \qedhere
\end{align*}
\end{proof}
%
\begin{lemma}
\label{lem:blockdet_gen}
For $i\in [\ell]$, let $C_i$ be an $a\times (a+t_i)$ matrix and $D_i$ be an $h \times (a+t_i)$ matrix for some $t_1+t_2+\cdots+t_\ell=h$ where $t_i\ge 1$. Then,
\begin{align*}
\det \left[
\begin{array}{c|c|c|c}
C_1 & 0 & \cdots & 0\\
\hline
0 &C_2 & \cdots & 0\\
\hline
\vdots& \vdots & \ddots & \vdots \\
\hline
0 & 0 & \cdots & C_\ell \\
\hline
D_1 & D_2 & \cdots & D_\ell \\
\end{array}
\right] = (-1)^{a(\sum_{i=1}^\ell t_i(\ell-i))}
\sum_{S_1\sqcup \cdots \sqcup S_\ell =[h], |S_i|=t_i} \sgn(S_1,\cdots,S_\ell) \prod_{i\in [\ell]} \det \mattwoone{C_i}{D_i^{(S_i)}}
\end{align*}
where $S_1\sqcup \dots \sqcup S_\ell$ ranges over partitions of $[h]$ such that $|S_i|=t_i$ and $\sgn(S_1,\cdots,S_\ell)$ is defined as in Lemma~\ref{lem:Laplace_expansion}.
\end{lemma}
\begin{proof}
Let
\begin{align*}
F=\left[F_1|F_2|\cdots|F_\ell\right]=
 \left[
\begin{array}{c|c|c|c}
C_1 & 0 & \cdots & 0\\
\hline
0 &C_2 & \cdots & 0\\
\hline
\vdots& \vdots & \ddots & \vdots \\
\hline
0 & 0 & \cdots & C_\ell \\
\hline
D_1 & D_2 & \cdots & D_\ell \\
\end{array}
\right].
\end{align*}
Let $[p,q]$ be the integers between $p$ and $q,$ i.e., $[p,q]=\{i:p\le i\le q\}$. By Lemma~\ref{lem:Laplace_expansion},
\begin{align*}
\det F =\det[F_1|F_2|\cdots|F_\ell]&=\sum_{T_1\sqcup \cdots \sqcup T_\ell =[a\ell+h], |T_i|=a+t_i} \sgn(T_1,\cdots,T_\ell) \prod_{i\in [\ell]} \det F_i^{(T_i)}
\end{align*}
Note that the only terms which survive correspond to partitions $T_1\sqcup T_2 \sqcup \cdots \sqcup T_\ell$ of rows of $F$ such that for every $i\in [\ell]$, $T_i$ contains the rows of $C_i$ (i.e. $[(i-1)a+1,ia]$). In the other terms, there exists some $i\in [\ell]$ such that $F_i^{(T_i)}$ contains a zero row and thus $\det F_i^{(T_i)}=0$. Such partitions are given by $T_i=[(i-1)a+1,ia]\cup S_i$ where $S_1\sqcup S_2 \cdots \sqcup S_\ell$ is some partition of rows of $[D_1|D_2|\cdots|D_\ell]$ such that $|S_i|=t_i$. So the expansion for $\det F$ can be written as:
\begin{align*}
\det F&=\sum_{S_1\sqcup \cdots \sqcup S_\ell =[a\ell+1,a\ell+h], |S_i|=t_i} \sgn([1,a]\cup S_1,\cdots,[(\ell-1)a+1,\ell a]\cup S_\ell) \prod_{i\in [\ell]} \det F_i^{([(i-1)a,ia]\cup S_i)}\\
&= (-1)^{a(\sum_{i=1}^\ell t_i(\ell-i))} \sum_{S_1\sqcup \cdots \sqcup S_\ell =[a\ell+1,a\ell+h], |S_i|=t_i} \sgn([1,\ell a],S_1,S_2,\cdots S_\ell) \prod_{i\in [\ell]} \det F_i^{([(i-1)a,ia]\cup S_i)}\\
 &=(-1)^{a(\sum_{i=1}^\ell t_i(\ell-i))} \sum_{S_1\sqcup \cdots \sqcup S_\ell =[h], |S_i|=t_i} \sgn(S_1,S_2,\cdots S_\ell) \prod_{i\in [\ell]} \det \mattwoone{C_i}{D_i^{(S_i)}}. \qedhere
\end{align*}

\end{proof}

\noindent We will now prove Lemma~\ref{lem:blockdet}, which was used in our constructions in Sections~\ref{Sec:H2} and \ref{Sec:H3}.
\begin{proof}[Proof of Lemma~\ref{lem:blockdet}]
After applying Lemma~\ref{lem:blockdet_gen}, we just need to note that 
$$\sum_{S_1\sqcup \cdots \sqcup S_h =[h], |S_i|=1} \sgn(S_1,\cdots,S_\ell) \prod_{i\in [\ell]} \det \mattwoone{C_i}{D_i^{(S_i)}}= \sum_{\pi} \sgn(\pi) \prod_{i\in [h]} \det \mattwoone{C_i}{D_i^{(\pi(i))}}$$ where the last summation is over all permutations $\pi$ of $h$ elements which is the exactly the required determinant.
\end{proof}

\section{Proof of Lemma~\ref{lem:subgroup}}
\label{Sec:subgroup_lemma}
The goal of the section is to prove Lemma~\ref{lem:subgroup} which is restated here for convenience.
\begin{lemma}[Restatement of Lemma~\ref{lem:subgroup}]
Let $r,n$ be some positive integers with $r\le n$. Then there exists a finite field $\F_q$ with $q=O(n)$ such that the multiplicative group $\F_q^*$ contains a subgroup of size at least $r$ and with at least $n/r$ cosets. If additionally we require that the field has characteristic two, then such a field exists with $q=n\cdot \exp(O(\sqrt{\log n})).$
\end{lemma}
\noindent
We will need some estimates from analytic number theory, we will setup some notation first.
\begin{align*}
&\pi(x;m,a): \text{number of primes $p\le x$ such that $p\equiv a \mod m$}\\
&\pi(x,y;m,a)=\pi(y;m,a)-\pi(x;m,a)\\
&\Li(x)=\int_2^x \frac{1}{\ln t}dt\\
&(m,a):\ \text{greatest common divisor of $m$ and $a$}\\
&\phi(m):\ \text{number of positive integers $a\le m$ such that $(a,m)=1$ (Euler's totient function)}
\end{align*}
By the prime number theorem, the number of primes $\le x$ is approximately $\Li(x)=\Theta(x/\log x)$. So if the primes are equidistributed among different congruence classes of $m$ with no obvious divisors (i.e. $a\mod m$ where $(a,m)=1$), then we expect to see approximately $\Li(x)/\phi(m)$ primes in each such congruence class. The following theorem gives an upper bound on the error term in this approximation averaged over $m< \sqrt{x}(\log x)^A$.

\begin{theorem}[Theorem from \cite{BFI86} (Page 250)]
\label{Thm:Bomb_Vinogradov}
Let $a\ne 0, A\ge 0$ be some fixed constants and $x\ge 3$. We then have
$$\sum\limits_{(m,a)=1;\ m<\sqrt{x}(\log x)^A}\left|\pi(x;m,a)-\frac{\Li(x)}{\phi(m)}\right| \lesssim_{a,A} x\frac{(\log\log x)^B}{(\log x)^3}$$
where $B$ is an absolute constant.
\end{theorem}

Applying the above theorem with $a=1,A=0$ for $x$ and $2x$, and using triangle inequality, we get the following corollary.
\begin{corollary}\label{cor:PNT_error}
For $x$ large enough,
$$\sum\limits_{m< \sqrt{x}}\left|\pi(x,2x;m,1)-\frac{(\Li(2x)-\Li(x))}{\phi(m)}\right| \lesssim x\frac{(\log\log x)^B}{(\log x)^3}$$
where $B$ is an absolute constant.
\end{corollary}

\begin{lemma}\label{lem:NT_lemma_prime}
Let $a\le b$ be some positive integers. Then there exists $A\ge a$,$B\ge b$ such that $AB+1$ is a prime and $AB=O(ab)$.
\end{lemma}
\begin{proof}
If there exists some $A$ such that $a\le A \le 2a$ and there is a prime $p$ between $4ab+1$ and $8ab$ which is congruent to $1 \mod A$, then we can take $B=(p-1)/A\ge b$. Suppose this is not true, we will arrive at a contradiction. For every $a\le m \le 2a$, we have $\pi(4ab,8ab;m,1)=0.$ Applying corollary~\ref{cor:PNT_error} with $x=4ab$, we get
\begin{align*}
ab\frac{(\log\log ab)^B}{(\log ab)^3} &\gtrsim \sum_{m< 2\sqrt{ab}} \left|\pi(4ab,8ab;m,1)-\frac{(\Li(8ab)-\Li(4ab))}{\phi(m)}\right|\\
&\ge \sum_{a\le m< 2a} \left|\pi(4ab,8ab;m,1)-\frac{(\Li(8ab)-\Li(4ab))}{\phi(m)}\right|\\
&=\sum_{a\le m< 2a} \frac{(\Li(8ab)-\Li(4ab))}{\phi(m)}\\
&\ge a \frac{\Li(8ab)-\Li(4ab)}{2a} \gtrsim \frac{ab}{\log(ab)}
\end{align*}
which is a contradiction when $ab$ is large enough.
\end{proof}

In practice, it is desirable to work with fields of characteristic two, the following lemma gives us such fields.

\begin{lemma}\label{lem:NT_lemma_two}
Let  $a,b$ be some positive integers and let $n=ab$. Then there exists $A\ge a$, $B\ge b$ such that $q=AB+1$ is a power of two and $q=n\cdot \exp(O\sqrt{\log n}).$
\end{lemma}
\begin{proof}
Let $m$ be a positive integer to be chosen later. Let $\ell$ be an integer such that $$2^{\ell(2^m-1)}\ge Cn+1 > 2^{(\ell-1)(2^m-1)}$$ where $C\ge 1$ is some sufficiently large constant to be chosen later and let $x=2^\ell, q=x^{2^m}$. We will now show that for any $a\le n$, we can factor $q-1$ as $A\cdot B$ where $A\ge a$ and $B\ge n/a=b$.
We can factor $q-1=x^{2^m}-1$ as: $$x^{2^m}-1= (x-1) \prod_{i\in [m]} (1+x^{2^{i-1}}).$$ We will rearrange these factors to get the desired factorization of $q-1$. Let $0\le \alpha\le 2^m-1$ be such that $x^{\alpha-1}< a\le x^\alpha$. Expand $\alpha$ into its binary expansion as $\alpha=\sum_{i\in S} 2^i$ where $S \subset \{0,1,\cdots,m-1\}$. Define $A=\prod_{i\in S}(1+x^{2^i})$ and define $B=(x^{2^m}-1)/A$. Clearly $A\ge x^\alpha \ge a$. We can lower bound $B$ as follows:
\begin{align*}
B&=\frac{(x^{2^m}-1)}{\prod_{i\in S}(1+x^{2^i})}= \prod_{i\in S} (1+x^{-2^i})^{-1} \cdot \frac{(x^{2^m}-1)}{\prod_{i\in S} x^{2^i}}\\
& \ge \exp(-\sum_{j\ge 0} x^{-2^{j}}) \frac{(x^{2^m}-1)}{x^\alpha} \ge \exp(-\sum_{j\ge 0} 2^{-2^{j}}) \frac{(x^{2^m}-1)}{xa}\\
& \ge \exp(-\sum_{j\ge 0} 2^{-2^{j}}) \frac{(x^{2^m-1}-1)}{a} \ge \exp(-\sum_{j\ge 0} 2^{-2^{j}}) \frac{Cn}{a} \ge \frac{n}{a}
\end{align*}
when $C=\exp(\sum_{j\ge 0} 2^{-2^{j}})$. Now we need to bound $q=x^{2^m}$ as a function of $n.$
\begin{align*}
q=2^{\ell 2^m}&= 2^{(\ell-1)(2^m-1)}\cdot 2^\ell \cdot 2^{2^m-1}\\
&\le (Cn+1)\cdot 2^{\ell} \cdot 2^{2^m-1}\\
& \lesssim n^{1+1/(2^m-1)} \cdot 2^{2^m-1}\\
&\lesssim n\exp(O(\sqrt{\log n}))
\end{align*}
if we choose $m$ such that $(2^m-1)=\Theta(\sqrt{\log n})$.

\end{proof}
\noindent
We are now ready to prove Lemma~\ref{lem:subgroup}.
\begin{proof} [Proof of Lemma~\ref{lem:subgroup}]
By Lemma~\ref{lem:NT_lemma_prime}, there exists $A\ge r$ and $B\ge n/r$ such that $q=AB+1$ is prime and $q=O(n)$. Since $\F_q^*$ is a cyclic group of size $q-1$ and $A$ divides $q-1$, there exists a subgroup of $\F_q^*$ of size $A\ge r$ with $B\ge n/r$ cosets. To get a finite field of characteristic two, we use Lemma~\ref{lem:NT_lemma_two} instead.
\end{proof}

\end{document}